\documentclass{amsart}

\usepackage{standalone}
\usepackage{tikz}
\usepackage{enumerate}
\usepackage{amsfonts,amsmath,amssymb,amsfonts}
\usepackage{stmaryrd}
\usepackage{subfig}
\usepackage{caption} 
\usepackage{xcomment}
\usepackage{xspace}
\usetikzlibrary{decorations.pathmorphing,decorations.pathreplacing,arrows,automata,patterns}
\usepackage{complexity}

\usepackage[style=alphabetic,firstinits=true,maxnames=6,minnames=5,
maxbibnames=99,natbib=true,
backend=bibtex]{ biblatex }
\newcommand{\ZZ}{\mathbb{Z}}

\newcommand{\NN}{\mathbb{N}}
\newcommand{\F}{\mathcal{F}}
\newcommand{\card}[1]{|#1|}
\newcommand{\vect}[1]{\mathbf{#1}}
\renewcommand{\vec}[1]{\vect{#1}}

\newcommand{\cne}{{\NE}}

\newcommand{\cnp}{{\NP}}
\newcommand{\cp}{{\P}}
\newcommand{\csharp}{{\#\P}}

\newcommand{\unaire}[1]{un(#1)}
\newcommand{\nspace}[1]{{\NSPACE(#1)}}
\newcommand{\ntime}[1]{{\NTIME(#1)}}

\newcommand{\ntimeun}[1]{{\NTIME_1(#1)}}

\newcommand{\ntimebin}[1]{{\NTIME_2(#1)}}

\newcommand{\strongper}[1]{{\mathfrak{P}_{#1}}}
\newcommand{\oneper}[1]{{\mathfrak{P}^1_{#1}}}

\newcommand{\horper}[1]{{\mathfrak{P}^h_{#1}}}
\newcommand{\nbper}[2]{{\mathfrak{N}_{#1}\ifthenelse{\equal{#2}{}}{}{(#2)}}}
\newcommand{\recper}[1]{{\mathfrak{P}^r_{#1}}}

\newtheorem{theorem}{Theorem}[section]

\newtheorem{lemma}{Lemma}[section]

\newcommand{\includepicture}[1]{\includegraphics{figures/images/#1.pdf}}

\newcommand{\tuilesMTbgcolor}{blue!20}
\newcommand{\etat}{$s$}
\newcommand{\ruban}{$a$}
\newcommand{\etatN}{$s'$}
\newcommand{\rubanN}{$a'$}

\tikzstyle{statestyle} = [draw,circle,inner sep=1pt]
\tikzstyle{rubanstyle} = [draw=black!30,thick,rectangle,fill=white]
\tikzstyle{transitionstyle} = [-latex]

\newcommand{\mtTransitionG}[1][false]{
\filldraw[fill=\tuilesMTbgcolor] (0,0) rectangle (4,4);
\node[statestyle] (etat) at (1,2) {\etat};
\draw[transitionstyle] (etat) -- (1,3) -- (0,3);
\draw[transitionstyle] (1,0) -- (etat);
\node[rubanstyle] (ruban) at (3,2) {\ruban};

\ifthenelse{\equal{#1}{false}}{}{
\node[left] at (0,3) (hautetat) {\etatN};
\node[above] at (2,4) (hautruban) {\rubanN};
\node[below] at (2,0) (basruban) {\ruban};
\node[below] at (1,0) (basetat) {\etat} ;
}
}

\newcommand{\mtTransitionR}[1][false]{
\filldraw[fill=\tuilesMTbgcolor] (0,0) rectangle (4,4);
\node[statestyle] (etat) at (1,2) {\etat};
\draw[transitionstyle] (1,0) -- (etat);
\draw[transitionstyle] (etat) -- (1,4);
\node[rubanstyle] (ruban) at (3,2) {\ruban};

\ifthenelse{\equal{#1}{false}}{}{
\node[above] at (1,4) (hautetat) {\etatN};
\node[above] at (2,4) (hautruban) {\rubanN};
\node[below] at (2,0) (basruban) {\ruban};
\node[below] at (1,0) (basetat) {\etat} ;
}
}

\newcommand{\mtTransitionD}[1][false]{
\filldraw[fill=\tuilesMTbgcolor] (0,0) rectangle (4,4);
\node[statestyle] (etat) at (1,2) {\etat};
\draw[transitionstyle] (1,0) -- (etat);
\draw[transitionstyle] (etat) -- (1,3) -- (4,3);
\node[rubanstyle] (ruban) at (3,2) {\ruban};

\ifthenelse{\equal{#1}{false}}{}{
\node[right] at (4,3) (hautetat) {\etatN};
\node[above] at (2,4) (hautruban) {\rubanN};
\node[below] at (2,0) (basruban) {\ruban};
\node[below] at (1,0) (basetat) {\etat} ;
}
}

\newcommand{\mtRubanTransitionG}[1][false]{
\filldraw[fill=\tuilesMTbgcolor] (0,0) rectangle (4,4);
\draw[transitionstyle]  (4,3) -- (1,3) -- (1,4);
\node[rubanstyle] (ruban) at (3,2) {\ruban};

\ifthenelse{\equal{#1}{false}}{}{
\node[above] at (1,4) (hautetat) {\etat};
\node[above] at (2,4) (hautruban) {\ruban};
\node[below] at (2,0) (basruban) {\ruban};
\node[right] at (4,3)  (basetat) {\etat};
}
}

\newcommand{\mtRubanTransitionD}[1][false]{
\fill[color=\tuilesMTbgcolor] (0,0) rectangle (4,4);
\draw (0,0) rectangle (4,4);
\draw[transitionstyle] (0,3) -- (1,3) -- (1,4);
\node[rubanstyle] (ruban) at (3,2) {\ruban};

\ifthenelse{\equal{#1}{false}}{}{
\node[above] at (1,4) (hautetat) {\etat};
\node[left] at (0,3) (basetat) {\etat};
\node[above] at (2,4) (hautruban) {\ruban};
}
}

\newcommand{\mtRuban}[1][false]{
\filldraw[fill=\tuilesMTbgcolor] (0,0) rectangle (4,4);
\node[rubanstyle] (ruban) at (3,2) {\ruban};

\ifthenelse{\equal{#1}{false}}{}{
\node[above] at (2,4) (hautruban) {\ruban};
\node[below] at (2,0) (basruban) {\ruban};
}
}

\newcommand{\mtTransitionInitD}[1][false]{
\fill[color=\tuilesMTbgcolor] (0,1) rectangle (4,4);
\draw (0,0) rectangle (4,4);
\node[statestyle] (etat) at (1,2) {\etat};
\draw[transitionstyle] (etat) -- (1,3) -- (4,3);
\draw[transitionstyle] (0,2) -- (etat);
\node[rubanstyle] (ruban) at (3,2) {\ruban};

\ifthenelse{\equal{#1}{false}}{}{
\node[left] at (0,2) (basetat) {\etat};
\node[right] at (4,3) (hautetat) {\etatN};
\node[above] at (2,4) (hautruban) {\rubanN};
}
}

\newcommand{\mtTransitionInitR}[1][false]{
\fill[color=\tuilesMTbgcolor] (0,1) rectangle (4,4);
\draw (0,0) rectangle (4,4);
\node[statestyle] (etat) at (1,2) {\etat};
\draw[transitionstyle] (etat) -- (1,4);
\draw[transitionstyle] (0,2) -- (etat);
\node[rubanstyle] (ruban) at (3,2) {\ruban};

\ifthenelse{\equal{#1}{false}}{}{
\node[above] at (1,4) (hautetat) {\etatN};
\node[above] at (2,4) (hautruban) {\rubanN};
\node[left] at (0,2) (basetat) {\etat};
}
}

\newcommand{\mtBordB}[1][false]{
\fill[color=\tuilesMTbgcolor] (0,1) rectangle (4,4);
\draw (0,0) rectangle (4,4);
\node[rubanstyle] (ruban) at (3,2) {\ruban};

\ifthenelse{\equal{#1}{false}}{}{
\node[above] at (2,4) (hautruban) {\ruban};
}
}

\newcommand{\mtBordBtrans}[1][false]{
\fill[color=\tuilesMTbgcolor] (0,1) rectangle (4,4);
\draw (0,0) rectangle (4,4);
\node[rubanstyle] (ruban) at (3,2) {\ruban};
\draw[transitionstyle] (0,3) -- (1,3) -- (1,4);

\ifthenelse{\equal{#1}{false}}{}{
\node[above] at (2,4) (hautruban) {\ruban};
\node[left] at (0,3) (gaucheetat) {\etat};
\node[above] at (1,4) (hautetat) {\etat};
}
}

\newcommand{\mtBordG}[1][false]{
\fill[color=\tuilesMTbgcolor] (4,0) rectangle (1,4);
\draw (0,0) rectangle (4,4);
}

\newcommand{\mtBordD}[1][false]{
\fill[color=\tuilesMTbgcolor] (0,0) rectangle (3,4);
\draw (0,0) rectangle (4,4);
}

\newcommand{\mtBordH}[1][false]{
\fill[color=\tuilesMTbgcolor] (0,0) rectangle (4,3);
\draw (0,0) rectangle (4,4);

\ifthenelse{\equal{#1}{false}}{}{
\node[below] at (2,0) (basruban) {\ruban};
}
}

\newcommand{\mtBordHh}[1][false]{
\fill[color=\tuilesMTbgcolor] (0,0) rectangle (4,3);
\node[statestyle] (etat) at (1,2) {\etat};
\draw (1,0) -- (etat);
\draw (0,0) rectangle (4,4);

\ifthenelse{\equal{#1}{false}}{}{
\node[below] at (1,0) (basetat) {\etat};
\node[below] at (2,0) (basruban) {\ruban};
}
}

\newcommand{\mtBordHG}[1][false]{
\fill[color=\tuilesMTbgcolor] (1,0) rectangle (4,3);
\draw (0,0) rectangle (4,4);
}
\newcommand{\mtBordHD}[1][false]{
\fill[color=\tuilesMTbgcolor] (0,0) rectangle (3,3);
\draw (0,0) rectangle (4,4);
}

\newcommand{\mtBordBD}[1][false]{
\fill[color=\tuilesMTbgcolor] (0,1) rectangle (3,4);
\draw (0,0) rectangle (4,4);
}

\newcommand{\mtBordBG}[1][false]{
\fill[color=\tuilesMTbgcolor] (4,1) rectangle (1,4);
\draw (0,0) rectangle (4,4);
\draw (2,2) -- (4,2);

\ifthenelse{\equal{#1}{false}}{}{
 \node[right] at (4,2) (basetat) {\etat};
}
}

\newcommand{\maindoubles}{red}

\newcommand{\underdoublessides}{red}
\newcommand{\underdoublescenter}{red!70}

\newcommand{\highlight}[1]{#1}
\newcommand{\cross}{
   \highlight{
    \fill[color=\maindoubles] (-0.5,2) rectangle (0,-0.5) rectangle (2,0);
   }
   \draw (-2,-2) rectangle +(4,4);
   \draw[latex-latex] (-2,0) -- +(4,0);
   \draw[latex-latex] (0,-2) -- +(0,4);
   \draw[latex-latex] (-2,-2) +(1.5,4) -- +(1.5,1.5) -- +(4,1.5);
}
\newcommand{\armoo}{
   \draw (-2,-2) rectangle +(4,4);
   \draw[-latex] (0,2) -- +(0,-4);
   \draw[-latex] (-2,0) -- +(1,0);
   \draw[-latex] (2,0) -- +(-1,0);
}
\newcommand{\armdro}{
   \draw (-2,-2) rectangle +(4,4);
   \highlight{
     \fill[color=\maindoubles] (0,2) rectangle (-0.5,-2);
   }
   \draw[-latex] (-2,-2) +(2,4) -- +(2,0);
   \draw[-latex] (-2,-2) +(1.5,4) -- +(1.5,0);
   \draw[-latex] (-2,0) -- +(1,0);
   \draw[-latex] (2,0) -- +(-1,0);
}
\newcommand{\armdrd}{
   \draw (-2,-2) rectangle +(4,4);
   \highlight{
   \fill[color=\underdoublescenter] (1,0) rectangle +(-2,-0.5);
   \fill[color=\maindoubles] (0,2) rectangle (-0.5,-2);
   \fill[color=\underdoublessides] (-2,0) rectangle +(1,-0.5);
   \fill[color=\underdoublessides] (2,0) rectangle +(-1,-0.5);
   }
   \draw[-latex] (-2,-2) +(2,4) -- +(2,0);
   \draw[-latex] (-2,-2) +(1.5,4) -- +(1.5,0);
   \draw[-latex] (-2,0) -- +(1,0);
   \draw[-latex] (2,0) -- +(-1,0);
   \draw[-latex] (-2,-0.5) -- +(1,0);
   \draw[-latex] (2,-0.5) -- +(-1,0);
}
\newcommand{\armdlo}{
   \draw (-2,-2) rectangle +(4,4);
   \highlight{
   \fill[color=\maindoubles] (0,2) rectangle (0.5,-2);
   }
   \draw[-latex] (-2,-2) +(2,4) -- +(2,0);
   \draw[-latex] (-2,-2) +(2.5,4) -- +(2.5,0);
   \draw[-latex] (-2,0) -- +(1,0);
   \draw[-latex] (2,0) -- +(-1,0);
}
\newcommand{\armdld}{
   \draw (-2,-2) rectangle +(4,4);
   \highlight{
   \fill[color=\underdoublescenter] (1,0) rectangle +(-2,-0.5);
   \fill[color=\maindoubles] (0,2) rectangle (0.5,-2);
   \fill[color=\underdoublessides] (-2,0) rectangle +(1,-0.5);
   \fill[color=\underdoublessides] (2,0) rectangle +(-1,-0.5);
   }
   \draw[-latex] (-2,-2) +(2,4) -- +(2,0);
   \draw[-latex] (-2,-2) +(2.5,4) -- +(2.5,0);
   \draw[-latex] (-2,0) -- +(1,0);
   \draw[-latex] (2,0) -- +(-1,0);
   \draw[-latex] (-2,-0.5) -- +(1,0);
   \draw[-latex] (2,-0.5) -- +(-1,0);
}
\newcommand{\armod}{
   \draw (-2,-2) rectangle +(4,4);
   \highlight{
   \fill[color=\underdoublescenter] (1,0) rectangle +(-2,-0.5);
   \fill[color=\underdoublessides] (-2,0) rectangle +(1,-0.5);
   \fill[color=\underdoublessides] (2,0) rectangle +(-1,-0.5);
   }
   \draw[-latex] (-2,-2) +(2,4) -- +(2,0);
   \draw[-latex] (-2,0) -- +(1,0);
   \draw[-latex] (2,0) -- +(-1,0);
   \draw[-latex] (-2,-0.5) -- +(1,0);
   \draw[-latex] (2,-0.5) -- +(-1,0);
}

\newcommand{\sbreaker}{\fill[color=black] (0.375,0) rectangle (0.625,1);}
\newcommand{\breaker}{\tikz[scale=0.25]{\sbreaker}\xspace}
\newcommand{\shorline}{\fill[color=black] (-0.5,0) rectangle (0.5,0.25);}
\newcommand{\horline}{\tikz[scale=0.25]{\shorline}\xspace}
\newcommand{\shempty}{\draw[color=black] (-0.5,0) -- (0.5,0);}
\newcommand{\hempty}{\tikz[scale=0.25]{\shempty}\xspace}
\newcommand{\scorner}{
  \fill[color=black] (0.375,0) rectangle (0.625,1);
  \fill[color=black] (0,0.375) rectangle (1,0.625);
}
\newcommand{\corner}{\tikz[scale=0.25]{\scorner}\xspace}

\colorlet{leftcolor}{yellow}
\colorlet{rightcolor}{red}
\colorlet{outsidecolor}{blue}

\newcommand{\sdiagdiag}{
\fill[color=leftcolor](0,0) -- (1,1) -- (0,1) -- cycle;
\fill[color=rightcolor](0,0) -- (1,1) -- (1,0) -- cycle;
}
\newcommand{\sdiagleft}{
\fill[color=leftcolor](0,0) rectangle  (1,1);
}
\newcommand{\sdiagright}{
\fill[color=rightcolor](0,0) rectangle (1,1);
}
\newcommand{\sdiaghor}{
\fill[color=leftcolor](0,0) rectangle (1,0.5);
\fill[color=rightcolor](0,0.5) rectangle (1,1);
}
\newcommand{\sdiagver}{
\fill[color=leftcolor](0.5,0) rectangle (1,1);
\fill[color=rightcolor](0,0) rectangle (0.5,1);
}
\newcommand{\sdiagcorner}{
\fill[color=rightcolor](0,0) rectangle (1,1);
\fill[color=leftcolor](0.5,0) rectangle (1,0.5)  (0.5,0.5) -- (1,1) -- (0.5,1) -- cycle (0,0) --
(0.5,0.5) -- (0,0.5) -- cycle;
}
\newcommand{\diagcorner}{\tikz[scale=0.25]{\sdiagcorner}\xspace}
\newcommand{\diagver}{\tikz[scale=0.25]{\sdiagver}\xspace}
\newcommand{\diaghor}{\tikz[scale=0.25]{\sdiaghor}\xspace}
\newcommand{\diagright}{\tikz[scale=0.25]{\sdiagright}\xspace}
\newcommand{\diagleft}{\tikz[scale=0.25]{\sdiagleft}\xspace}
\newcommand{\diagdiag}{\tikz[scale=0.25]{\sdiagdiag}\xspace}

\colorlet{gridcolor}{black!60}
\colorlet{myyellow}{yellow}
\colorlet{myblue}{blue!40}
\colorlet{mygray1}{black!15}
\colorlet{mygray2}{black!75}
\colorlet{colorleft}{red}
\colorlet{colorright}{yellow}

\newcommand{\tnoir}{\tikz[scale=0.25]{\stnoir}\xspace}
\newcommand{\rightmost}{\tikz[scale=0.25]{\srightmost}\xspace}
\newcommand{\leftmost}{\tikz[scale=0.25]{\sleftmost}\xspace}
\newcommand{\betweenrl}{\tikz[scale=0.25]{\sbetweenrl}\xspace}
\newcommand{\betweenlr}{\tikz[scale=0.25]{\sbetweenlr}\xspace}

\newcommand{\stnoir}{
  \draw[fill=black] (1,1) -- (0,0) -- (0,1) -- (1,0) --cycle;
\draw[color=gridcolor] (0,0) rectangle (1,1);
}

\newcommand{\srightmost}{
\filldraw[fill=black] (0,0) -- (0.5,0.5) -- (0,1) -- cycle;
\filldraw[fill=mygray2] (0,0) -- (0.5,0.5) -- (1,0) --cycle;
\filldraw[fill=mygray1] (0,1) -- (0.5,0.5) -- (1,1) -- cycle;
\draw[color=gridcolor] (1,1) rectangle (0,0);
}

\newcommand{\sleftmost}{
\filldraw[fill=black] (1,0) -- (0.5,0.5) -- (1,1) -- cycle;
\filldraw[fill=mygray2] (0,1) -- (0.5,0.5) -- (1,1) -- cycle;
\filldraw[fill=mygray1] (0,0) -- (0.5,0.5) -- (1,0) --cycle;
\draw[color=gridcolor] (1,1) rectangle (0,0);
}

\newcommand{\sbetweenrl}{
  \draw[color=black,fill=mygray2] (0,0) -- (1,1) -- (0,1) -- (1,0) -- cycle;
  \draw[color=gridcolor] (0,0) rectangle (1,1);
}

\newcommand{\sbetweenlr}{  
  \draw[color=black,fill=mygray1] (0,0) -- (1,1) -- (0,1) -- (1,0) -- cycle;
  \draw[color=gridcolor] (0,0) rectangle (1,1);
}

\newcommand{\trsdiagg}{\scalebox{0.25}{\tikz{\strsdiagg}}\xspace}

\newcommand{\strsdiagg}{
\filldraw[fill=gray] (0,0) rectangle (1,1);
\draw[line width=3pt,-latex] (0,0) -- (1,1);
}

\newcommand{\srectvertleft}{
 \fill[color=outsidecolor] (0,0) rectangle (0.5,1);
 \fill[color=leftcolor] (0.5,0) rectangle (1,1); 
 \draw[color=gridcolor] (0,0) rectangle (1,1);
}

\newcommand{\srectvertright}{
 \fill[color=outsidecolor] (0.5,0) rectangle (1,1);
 \fill[color=rightcolor] (0,0) rectangle (0.5,1); 
 \draw[color=gridcolor] (0,0) rectangle (1,1);
}

\newcommand{\srectcornerright}{
 \fill[color=outsidecolor] (0.5,0) rectangle (1,1);
 \fill[color=rightcolor] (0,0.5) rectangle (0.5,1); 
 \fill[color=rightcolor] (0,0) -- (0.5,0.5) -- (0.5,0) -- cycle; 
 \fill[color=leftcolor] (0,0) -- (0.5,0.5) -- (0,0.5) -- cycle; 
 \draw[color=gridcolor] (0,0) rectangle (1,1);
}

\newcommand{\srectcornerleft}{
 \fill[color=outsidecolor] (0,0) rectangle (0.5,1);
 \fill[color=leftcolor] (0.5,0) rectangle (1,0.5); 
 \fill[color=rightcolor] (0.5,0.5) -- (1,1) -- (1,0.5) -- cycle; 
 \fill[color=leftcolor] (0.5,0.5) -- (1,1) -- (0.5,1) -- cycle; 
 \draw[color=gridcolor] (0,0) rectangle (1,1);
}

\newcommand{\srecthoriz}{
 \fill[color=leftcolor] (0,0) rectangle +(1,0.5); 
 \fill[color=rightcolor] (0,0.5) rectangle +(1,0.5); 
 \draw[color=gridcolor] (0,0) rectangle (1,1);
}

\newcommand{\srectdiag}{
 \fill[color=leftcolor] (0,0) -- (1,1) -- (0,1) -- cycle; 
 \fill[color=rightcolor] (0,0) -- (1,1) -- (1,0) -- cycle; 
 \draw[color=gridcolor] (0,0) rectangle (1,1);
}

\newcommand{\srectright}{
 \fill[color=rightcolor] (0,0) rectangle +(1,1); 
 \draw[color=gridcolor] (0,0) rectangle (1,1); 
}

\newcommand{\srectleft}{
 \fill[color=leftcolor] (0,0) rectangle +(1,1); 
 \draw[color=gridcolor] (0,0) rectangle (1,1); 
}

\newcommand{\srectoutside}{
 \draw[color=gridcolor,fill=outsidecolor] (0,0) rectangle (1,1); 
}

\newcommand{\rectvertright}{\scalebox{0.25}{\tikz{\srectvertright}}\xspace}
\newcommand{\rectvertleft}{\scalebox{0.25}{\tikz{\srectvertleft}}\xspace}

\newcommand{\recthoriz}{\scalebox{0.25}{\tikz{\srecthoriz}}\xspace}
\newcommand{\rectcornerright}{\scalebox{0.25}{\tikz{\srectcornerright}}\xspace}
\newcommand{\rectcornerleft}{\scalebox{0.25}{\tikz{\srectcornerleft}}\xspace}
\newcommand{\rectdiag}{\scalebox{0.25}{\tikz{\srectdiag}}\xspace}
\newcommand{\rectright}{\scalebox{0.25}{\tikz{\srectright}}\xspace}
\newcommand{\rectleft}{\scalebox{0.25}{\tikz{\srectleft}}\xspace}
\newcommand{\rectoutside}{\scalebox{0.25}{\tikz{\srectoutside}}\xspace}

\begin{document}

\title{Characterizations of periods of multidimensional shifts}
\author{Emmanuel Jeandel and Pascal Vanier}
\address{LIF - Marseille}
\email{emmanuel.jeandel@lif.univ-mrs.fr}
\email{pascal.vanier@lif.univ-mrs.fr}

\begin{abstract}
We show that the sets of periods of multidimensional shifts of finite type 
(SFTs) are exactly the sets of integers of the complexity class $\NE$. We also 
show that the functions counting their number are the functions of $\#\E$. We 
also give characterizations of some other notions of periodicity. We finish the 
paper by giving some characterizations for sofic and effective subshifts.
\end{abstract}

\maketitle
\section{Introduction}

A multidimensional shift of finite type (SFT) is a set of colorings of
$\ZZ^2$ given by local rules. SFTs are one of the most fundamental
objects in symbolic dynamics \cite{LindMarcus}. 
One important question is to determine whether two SFTs are \emph{conjugate}. One approach to solve the
problem is by the study of invariants, i.e. quantities related to
subshifts that stay invariant under conjugacy.
The most significant invariants for SFTs are the entropy, that
measures the growth of the number of valid patterns, and the set of
periodic points.

These two invariants are relatively well known in the one-dimensional
case: The entropy of an SFT is the logarithm of the spectral radius of
some matrix related to the SFT, and the set of periodic points relates
to the cycles of the multi-graph represented by the matrix. As a
consequence, the set of integers $n$ so that there exists a periodic
point of period $n$ is a semi-linear set.

The situation becomes more complex for multi-dimensional SFTs.
This is linked to the fact that the emptiness problem for SFTs
becomes undecidable \cite{Berger2,Robinson}, 
in part due to the existence of nonempty SFTs
with no periodic points. While the theory of one-dimensional SFTs indeed
relates to the theory of finite automata of computer science, many
properties of multidimensional SFTs are to be understood using
\emph{computability theory}.

The case of the entropy was recently solved by Hochman and
Meyerovitch\cite{HochMey}: A real $\lambda \geq 0$ is the entropy of a
multidimensional SFT if and only if it is right approximable, that is
if we can compute a sequence of rational numbers converging to $\lambda$ from
above. This was one of the first articles in a series linking
dynamical properties of multidimensional SFTs and computability theory
\cite{MEdv,AubrunS09}.

In this article, we give a characterization of the sets of periods of
multidimensional SFTs using \emph{complexity theory}.
For a given point $x$ in an SFT, let $\Gamma_x = \left\{ \vect v\in \ZZ^2 \mid 
\forall \vect z \in
  \ZZ^2, x(\vect z+\vect v) = x(\vect z)\right\}$ be the lattice of periods of 
$x$.

We will study different notions of periodic points:
\begin{itemize}
	\item $c$ is \emph{strongly periodic} of period $n> 0$ if
	  $\Gamma_c = n\ZZ^2$
	\item $c$ is \emph{1-periodic} of period $\vect v \in \ZZ\times
	  \mathbb{N}  \setminus \{(0,0)\}$
	  if $\Gamma_c = \vect v\ZZ$
	\item $c$ is \emph{horizontally periodic} of period $n > 0$ if
	  $n$ is the least positive integer so that $n\ZZ \times
	  \{0\} \subseteq \Gamma_c$
\end{itemize}	
All these notions can readily be generalized for any dimension $d > 2$.
For a given subshift $X$, let $\strongper{X}$ (resp. $\oneper{X}, \horper{X}$)
denote the set of strong periods (resp. 1-periods, horizontal periods)
of $X$. The third notion seems a bit more peculiar. It is introduced
in this paper as a first, somewhat easier, result on which all other results
will be built.

To give a characterization of these sets in terms of complexity
classes, we will have to see these sets as languages.
If $L \subseteq \NN$, denote by $un(L) = \{ a^n | n \in L\}$.
If $L \subseteq \ZZ\times\NN$, denote
$\unaire{L} = \{ a^p b^q | (p,q) \in L\} \cup \{ a^p c^q| (-p,q) \in L\}$.

We will prove:
\begin{theorem}
	\label{thm:strong}
	For any $L \subseteq \NN$, there exists an SFT $X$ such that $L=\strongper{X}$ if and only if $\unaire{L} \in \cnp$.
\end{theorem} 
\begin{theorem}
	\label{thm:one}	
	For any $L \subseteq \ZZ \times \NN$, there exists a two dimensional SFT $X$ such that $L=\oneper{X}$
	if and only if $\unaire{L} \in \nspace{n}$.
	\end{theorem}
\begin{theorem}  
	\label{thm:hori}
	For any $L \subseteq \NN$, there exists a two dimensional SFT $X$ such that $L=\horper{X}$
	if and only if $\unaire{L} \in \nspace{n}$.
\end{theorem}

Here $\cnp$ denotes as usual the class of languages computable by a
nondeterministic Turing Machine in polynomial time. $\nspace{n}$
denotes the class of languages computable by a nondeterministic
Turing Machine in linear space.

Please note the slight difference between theorem~\ref{thm:strong} and the
others:
the two other theorems are valid for a fixed dimension $d = 2$.
Theorem~\ref{thm:strong} needs however to be formulated for all dimensions at
once: given
a language $L \in \cnp$ the dimension of the SFT for which $L$ is a set
of strong periods depends on $L$. It is in fact hard to provide a statement
valid exactly in dimension $d=2$. Intuitively, the reason is that SFTs can be
seen as a model of computation. For most models of
computation, the space complexity of a problem $L$ is generally the
same. However the time complexity of a problem depends on the exact
definition of the model: the problem to decide if a word is a palindrome is provably
quadratic in a model of a Turing Machine with one tape, but becomes
linear if the Turing Machine has two tapes. The two-dimensional SFTs
being assimilated as a new model of computation, there is no reason
for it to behave like a specific, already known, model of Turing
Machine. That's why we have to use a more robust class, $\cnp$, which
coincides for all reasonable models of computation (See e.g. the \emph{Invariance
Thesis}~\cite{BOAS}). The problem does not appear for space classes, as they are
already robust.

Rather than the set of periodic points, another interesting quantity
is the \emph{number} of periodic points. This quantity makes only
sense for strongly periodic points.
If $X$ is an SFT, denote by $\nbper{X}{}$ the map from $\{a\}^\star$ to $\NN$
that maps $a^n$ to the number of points of strong period $n$. Then we will prove:
\begin{theorem}\label{thm:strong:counting}
	Let $f:\{a\}^\star\rightarrow \NN$, there exists an SFT $X$ such that $=\nbper{X}{}$ 
	if and only if $f \in \csharp$
\end{theorem}

Here $\csharp$ denotes as usual the class of functions corresponding to the
number of accepting paths of a nondeterministic Turing Machine working
in polynomial time. This theorem gives a first insight on the
behavior of the Zeta function of a multidimensional SFT \cite{Lind}.

These characterizations in terms of complexity classes lead to some
closure properties on the set of periods. $\cnp$ and $\nspace{n}$ are
closed under intersection and union, so the sets of periods are also
closed under intersection and union. The closure by union could of
course be proven directly, by taking the disjoint union of the two SFTs.
Note that the closure by intersection is not trivial, as the classical construction
by Cartesian product does not work as usefully as it would seem: It is not true
that the set of strong periods of $X \times Y$ is the intersection of the sets of strong periods
of $X$ and $Y$: If $\{2\}$ (resp. $\{3\}$) 
is the set of strong periods of $X$ (resp. $Y$), then
the set of strong periods of $X\times Y$ will be $\{6\}$ and not
$\emptyset$ as intended. Note in particular that due to the peculiar
form of theorem \ref{thm:strong}, the SFT that realizes the intersection
of the strong periods of $X$ and $Y$ can be of higher dimension
than $X$ and $Y$. Finally, since nondeterministic space is closed under
complementation \cite{Immerman,BDG2}, 
1-periods are closed under complementation.
The question whether the sets of strong periods are closed under
complementation is of course related to the $\cp$ vs $\cnp$ problem
(More accurately, it is related to the question
$\cne=\textbf{co}\cne$, and to Asser's Problem \cite{JonesSelman,years}).

This paper is organized as follows: The first section gives the
necessary background both in multidimensional symbolic dynamics and
complexity theory.
We then proceed to prove all four theorems. We first prove 
Theorem~\ref{thm:hori} on horizontal periods. The techniques used for
this proof are the core of this article. The other three theorems
then build on this first proof, adding more and more complex structures
in the various constructions. We end this paper with a discussion on
similar results for multidimensional sofic and effective shifts rather than
SFTs.

Some of the results of this paper were announced at the DLT conference in the 
extended abstract~\cite{JV2010}.
\section{Preliminaries}
\subsection{Symbolic Dynamics}
We give here a primer on Multidimensional Symbolic Dynamics. See 
\cite{LindMarcus,Lind2} for more information.
\subsubsection{Subshifts}
Let $\Sigma$ be a finite alphabet and $d> 0$ an integer.
A \emph{configuration} over $\Sigma$ is a coloring of $\ZZ^d$ by $\Sigma$, 
that is a map : $\ZZ^d \rightarrow \Sigma$.
We denote by $\Sigma^{\ZZ^d}$ the set of all configurations over $\Sigma$,
$\Sigma^{\ZZ^d}$ is also called the \emph{full shift} on $\Sigma$.
A \emph{pattern} $P$ is a coloring of a subset $D\subset \ZZ^d$. $D$ is
the \emph{support} of the pattern. A pattern is finite if $D$ is finite.

A pattern $P$ of support $D$ \emph{appears} in a pattern $P'$ of
support $D'$
if there exists a position $\vect v \in D'$ so that $\vect v + D \subseteq D'$ and
$P(\vect v+\vect z) = P'(\vect z)$  for all $\vect z \in D$. 
We will write $P \in P'$ to say that $P$ appears in $P'$.

Let $\mathcal F$ be a set of \emph{finite} patterns.
A pattern is \emph{admissible} for $\mathcal F$ if it contains no patterns of $\mathcal F$.
The \emph{subshift} $X_{\mathcal F}$ defined by $\mathcal F$ is the set of all configurations where no pattern
of $\mathcal F$ appears:
\[ X_{\mathcal F} = \{ c \in \Sigma^{\ZZ^d} | \forall P \in {\mathcal F},  P\not\in c\} \]
A set $X$ is a subshift if there exists a set ${\mathcal F}$ so that $X = X_{\mathcal F}$.
Subshifts can also be characterized by a topological property, but we
will not need it here. 

A \emph{subshift of finite type} (or shortly SFT) is a subshift
$X_{\mathcal F}$ where $\mathcal F$ is finite. In this case, we can assume
that all patterns of $\mathcal F$ are over the same finite support $D$. In this
setting, a configuration $c$ is a \emph{valid} if all patterns of support $D$
appearing in $c$ are not in $\mathcal F$, such a configuration is called a point of $X$. 
The \emph{radius} of $D$ is the smallest $r$ so that $D \subseteq [-r,r]^d$. 
The radius of $X_{\mathcal F}$ is the radius of $D$. 

An \emph{effective subshift} is a subshift $X_{\mathcal F}$ where $\mathcal F$ 
is recursively enumerable.

Let $X$ and $Y$ be two $d$-dimensional subshifts, a \emph{block code} is a map 
$F:X\to Y$ such that there exists a map $f:{\Sigma_X^V}\to\Sigma_Y$, with 
$V=\{\vect{v_1},\dots,\vect{v_k}\}$ a finite subset of $\ZZ^2$ such that for 
any $\vect z\in\ZZ$:
\[
 F(x)_{\vect{z}}=f(x_{\vect z+\vect{v_1}},\dots,x_{\vect{z}+\vect{v_k}})
\]
A map $F:X\to Y$ is a factor map if it is surjective block code, $Y$ is then 
called a \emph{factor} of $X$. A subshift is called \emph{sofic} if it is a 
factor of some SFT.

\subsubsection{\label{Ss:perpoints}Periodic points}

Let $c\in \Sigma^{\ZZ^d}$ be a configuration. A vector $\vect v \in \ZZ^d$ is a
\emph{vector of periodicity} for $c$ if $c(\vect z) = c(\vect z+\vect v)$ for all $\vect z\in \ZZ^d$.
We write $\Gamma_c = \{ \vect v \in \ZZ^d | \forall \vect z \in \ZZ^d, c(\vect z) = c(\vect z +\vect v)\}$
for the set of vectors of periodicity of $c$. $\Gamma_c$ is of course a
lattice (i.e. a (discrete) subgroup of $\ZZ^d$).
There are three cases for $\Gamma_c$:
\begin{itemize}
	\item $\Gamma_c = \{0\}$: $c$ has no vector of periodicity.
	\item $\Gamma_c$ has rank $d$, then $c$ is \emph{periodic}. This corresponds to the notion of periodicity
	in dimension $1$: a finite orbit. In particular, in this case, one can prove that there exists $n$ such 
	that $n \ZZ^d \subseteq \Gamma_c$. If $n \ZZ^d= \Gamma_c$ we will say that $c$ is \emph{strongly periodic} of
	  strong period $n$. 
	\item $\Gamma_c$ has an intermediate rank $1 \leq k < d$. We will say
	  that $c$ is \emph{$k$-periodic}. 
\end{itemize}	
If $c$ is $1$-periodic, then $\gamma_c = v\mathbb{Z}$ for a unique vector $v$ (upto opposite direction)
which is called the \emph{$1$-period} of $c$.

\begin{figure}[htbp]
	\begin{center}
	\begin{tikzpicture}[scale=0.3]
	\clip (0,3) rectangle (20,20);
	\draw[gray] (0,0) grid (20,20);
	\draw[->,very thick] (10,0) -- (10,20);
	\draw[->,very thick] (0,10) -- (20,10);
	\draw (0,6) -- (20,14) ;
	\draw (0,10) -- (20,18);
	\fill[color=green!40, pattern=crosshatch dots gray,opacity=0.3] (0,6) -- (20,14) -- (20,18) -- (0,10);
	\draw node at (6,11) {$S_{m,n}$};
	\draw (0,14) -- (20,22) ;
	\fill[color=green!40, pattern=dots,opacity=0.3] (0,10) -- (20,18) -- (20,22) -- (0,14);
	\draw node at (3.5,13.3) {$S_{m,n} + (0,4r)$};
	\draw (0,2) -- (20,10) ;
	\fill[color=green!40, pattern=dots,opacity=0.3] (0,2) -- (20,10) --
	(20,14) -- (0,6);
	\draw node at (3.5,5.3) {$S_{m,n} - (0,4r)$};	
	\end{tikzpicture}
	\end{center}
	\caption{Representation of $S_{m,n}$ for $m=5, n=2$ and $r=1$.
Note that for every point $x$ in $S_{m,n}$, the square of size $2r$ and of
bottom right corner $x$ is entirely contained in $S_{m,n} \cup (S_{m,n} + (0,4r)) = D_{m,n}$.}
	\label{fig:snm}
\end{figure}
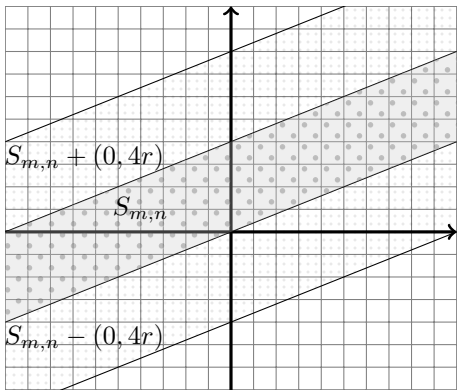

A configuration which is $1$-periodic in a subshift of dimension $d$ can be
seen as a configuration in a subshift of dimension $d-1$. To make this
statement exact, we need some definitions.
Let $X$ be an SFT of radius $r$.
Let $(m,n)$ be integers with $m \geq n \geq 0$.
Consider $S_{m,n} = \{ (x,y) \in \mathbb{Z}^2\ | 0 \leq -nx+my < 4mr \}$.
and $D_{m,n} = \{ (x,y) \in \mathbb{Z}^2\ | 0 \leq -nx+my < 8mr \} =
S_{m,n} \cup (S_{m,n} + (0,4r))$ represented in figure~\ref{fig:snm}.

Now the number of patterns of support $S_{m,n}$ that are $(m,n)$-periodic is
finite, as any such pattern is entirely determined by $S \cap \left( (0,m-1) \times \ZZ\right)$ which
contains at most $4rm$ points. So there are at most $|\Sigma|^{4rm}$
such patterns.

Now consider the following directed graph $G_{m,n}(X)$:
\begin{itemize}
	\item The vertices of $G_{m,n}(X)$ are all patterns of support $S_{m,n}$ that
	  are $(m,n)$ periodic.
	\item There is a edge from $P$ to $P'$ if
	the pattern $P \ominus P'$ of support
	$D_{m,n}$ defined by $P \ominus P' (z) = P(z)$ if $z \in S_{m,n}$ and $P \ominus P'(z) =
	P'(z-(0,4r))$ otherwise is valid for $X$.
\end{itemize}	
Now it is clear that there exists a bijection between configurations of $X$
with $(m,n)$ as a period and  bi-infinite paths in $G_{m,n}(X)$.
It is due to the fact that $\ZZ^2 = \uplus_i (S_{m,n} + (0,4r)i)$ and that
any square of size $2r$ whose bottom right corner is in $S_{m,n}$ is in
$D_{m,n}$ (this is where the hypothesis $m \geq n$ is used) so that any square of size $2r$ in $\ZZ^2$ is contained
in $D_{m,n} +(0,4r)i$ for some $i$.

Thus the entire information is contained in this graph. It is then easy to
obtain the following consequences:
\begin{lemma}\label{lem:horpergraph}
	Let $X$ be a two-dimensional SFT of radius $r$.
\begin{itemize}

	\item 
	If $X$ contains a configuration which is periodic of period $(m,n)$,
	then it contains a configuration which is fully periodic.
	(If the finite graph $G_{m,n}(X)$ contains an infinite path, it contains a cycle)
	\item
	More precisely, if $X$ contains a configuration with
	$(m,n)$ as a period, ${m \geq n >  0}$, then it contains a fully periodic
	configuration with $(m,n)$ and $(0,p)$ as periods, for some $0 < p \leq |\Sigma|^{4rm}$.
\end{itemize}
\end{lemma}

We could also obtain a characterization of $1$-periods with the graph:
\begin{lemma}\label{lem:onepergraph}
 $X$ admits $(m,n)$ as a $1$-period if and only if the graph   $G_{m,n}(X)$
 contains a path $u_0 \dots u_k$ so that
	  \begin{itemize}
		  \item $u_i=u_0$ for some $i <k$.
		  \item $u_{i+1} \not= u_1$
		  \item $u_j = u_k$ for some $i \leq j < k$
		  \item For each $(m',n')$ so that $\exists d\in\NN,(m,n) = d\cdot(m',n')$, 
                  there exists some pattern $u_l$ which is not $(m',n')$ periodic.
	  \end{itemize}
\end{lemma}	
We can of course choose $k \leq 3m|\Sigma|^{4rm}$.
The first three conditions ensure that there exists a configuration $c$ which is
$1$-periodic and admits $(m,n)$ as a period. The last condition ensures that
$(m,n)$ is indeed the least period of $c$.

\subsection{Aperiodic SFTs and determinism}\label{Ss:aper}

Let $X$ be a $\ZZ^d$ SFT, $X$ is \emph{aperiodic} when no point of $X$ admits a
periodicity vector. There are several well known such SFTs for $\ZZ^2$, most of them
come from the study of tilings, see e.g. Berger~\cite{Berger2,BergerPhd}, Robinson~\cite{Robinson},
Kari~\cite{Kari14,KariNil}. We will need two dimensional aperiodic SFTs with a particular property in this paper : 
determinism. A two dimensional SFT is \emph{north-west deterministic} if for
any two symbols $a,b$ at positions $(i,j)$ and $(i+1,j+1)$ there is at most
one symbol $c$ allowed at position $(i+1,j)$. Such an SFT was constructed by
Kari~\cite{KariNil}, his particular SFT will be used in 
section~\ref{S:counting}. \emph{East-determinism} can be determined in the same 
way: for any two symbols $a,b$ at positions $(i,j)$ and $(i,j+1)$ there is at 
most one symbol $c$ allowed at position $(i+1,j)$.

\subsection{Computational Complexity}

In this section we provide some background on computational complexity and its 
links with subshifts of finite type. More information about computational 
complexity and computability can be found in \cite{BDG,AroraBarak,Rogers}. 

Usually to model computation, Turing machines are used. Despite its power, this 
model is quite simple to describe : shortly, a Turing machine is a device with 
finite memory but that can read/write on an infinite tape at the position of 
its ``head''. Formally, it is a tuple $(Q,\Gamma,B,\Sigma,\delta,q_0,H)$ where:
\begin{itemize}
  \item $Q$ is a finite set of states,
  \item $\Gamma$ is the tape alphabet, the finite set of symbols that can appear on the tape,
  \item $B\in\Gamma$ is the ``Blank'' symbol,
  \item $\Sigma\subseteq\Gamma\setminus\{B\}$ is the input alphabet,
  \item $q_0$ is the initial state,
  \item $H$ is the set of halting states,
  \item $\delta:Q\setminus F\times\Gamma\rightarrow Q\times\Gamma\times\left\{\leftarrow,\cdot,\rightarrow\right\}$ 
    is the transition function. The symbols $\rightarrow$,$\leftarrow$ and $\cdot$ stand 
     for moving the head to the right, left and to let it where it is respectively.
\end{itemize}

If a problem can be answered by a Turing machine, then it is called 
\emph{decidable} and otherwise \emph{undecidable}. The most famous undecidable  
problem is the Halting Problem: deciding whether a given Turing machine halts 
with itself as an input. Another well known undecidable problem is the Domino 
Problem: given a set of Wang tiles, does it tile the plane?

A problem is called \emph{recursively enumerable} (r.e.) if there exists a 
Turing machine enumerating its elements and \emph{co-recursively enumerable} 
(co-r.e.) when its complement is r.e.

A complexity class is a class of problems decidable by a Turing machine such 
that some resource is bounded. The usual restrictions of the resources are on 
time or space : 
\begin{itemize}
 \item  the classes $\TIME(f(n))$ are the classes of problems decidable in time 
$f(n)$, where $n$ is the size of the input and  $f$ a fonction from $\NN$ to 
$\NN$.
 \item  the classes $\SPACE(f(n))$ are the classes of problems decidable in 
space $f(n)$.
\end{itemize}

Turing machines may be nondeterministic, which means that the transition 
function is multivalued. In this case, the input  is accepted if there exists a 
sequence of transitions leading to an accepting state. Time and space complexity 
 classes are also defined in the case of nondeterministic Turing machines :
\begin{itemize}
 \item  the classes $\NTIME(f(n))$ are the classes of problems 
nondeterministically decidable in time $f(n)$, where $n$ is the size of the 
input and  $f$ a fonction from $\NN$ to $\NN$.
 \item  the classes $\NSPACE(f(n))$ are the classes of problems 
nondeterministically decidable in space $f(n)$ 
\end{itemize}

\begin{figure}
\begin{center}
 \includepicture{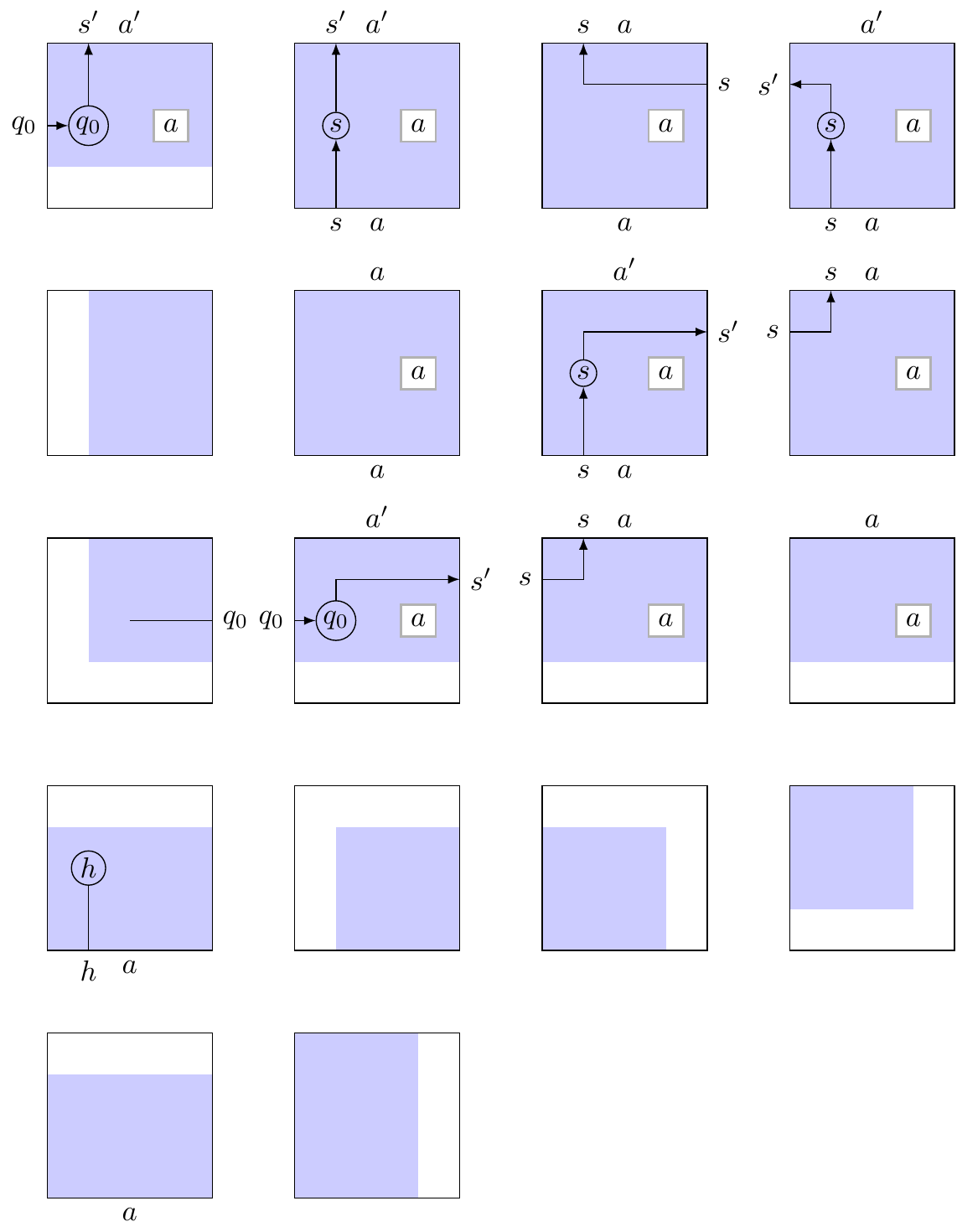}
\end{center}
\captionsetup{listof=false,singlelinecheck=off}
 \caption{A tiling system, given by Wang tiles, simulating a Turing machine. 
 The meaning of the labels are the following: \\
 \begin{itemize}
     \renewcommand{\labelitemi}{--}%
\item label $s_0$ represents the initial state of the Turing machine.
\item The top-left tile corresponds to the case where the Turing machine, given the state $s$ and the letter $a$ on the tape, writes $a'$,
moves the head to the left and to change from state $s$ to $s'$. The two
other tiles are similar.
\item $h$ represents a halting state. Note that the only states that can
  appear in the last step of a computation (before a border appears)
are halting states.
\end{itemize}}
\label{fig:mt_tuiles}
\end{figure}

As said earlier, tilings and recursivity are intimately linked.
In fact, it is quite easy to encode Turing machines in tilings. Such encodings can be found e.g. in \cite{KariRevCA,Chaitin08}. 
Given a Turing machine $M$, we can build a tiling system $\tau_M$
in figure~\ref{fig:mt_tuiles}. The tiling system is given by \emph{Wang tiles}, i.e., we can
only glue two tiles together if they coincide on their common edge.
This tiling system $\tau_M$ has the following property: there is an accepting path
for the word $u$ in time (less than) $t$ using space (less than) $w$ if and only if we can tile a
rectangle of size $(w+2) \times t$  with white borders, the  first row containing the input.
Note that this method works for both deterministic and nondeterministic machines.

\section{Horizontal periodicity in SFTs of dimension 2 and space complexity}

In this section, we give a proof for theorem~\ref{thm:hori}. We first prove that the unary language corresponding 
to horizontal periods can be recognized in 
linear space by a nondeterministic Turing machine (lemma~\ref{lem:hori:direct}) and then 
the reciprocal (lemma~\ref{lem:hori:reciproque}).

\begin{lemma}\label{lem:hori:direct}
Let $L\subseteq\NN^*$ be the set of horizontal periods of a two-dimensional SFT $X$,
then $\unaire{L}\in\nspace{n}$.
\end{lemma}
\begin{proof}
	Let $X = X_\F$ be a $2$-dimensional SFT on the alphabet $\Sigma$.
We will construct a nondeterministic Turing machine accepting $1^n$ if and only if $n+1$ is a
horizontal period of $X$. The machine has to work in space $\mathcal O(n)$, the input being 
given in unary.

Let $r$ be the radius of $X$, a point is in $X$ if and only if all its $r \times r$ 
blocks have no sub-pattern contained in $\F$.

Furthermore, we can prove that if there exists a point of horizontal period $n$,
then there also exists such a point, with vertical period at most $\card{\Sigma}^{rn}$.
Here is an algorithm, starting from $n$ as an input that checks whether $n$ is a horizontal
period of some point of $X$:
\begin{itemize}
	\item Initialize an array $P$ of size $n$ so that $P[i] = 1$ for all $i$.
	\item First choose nondeterministically $p \leq \card{\Sigma}^{rn}$
	\item Choose 
	  $r$ bi-infinite rows $(c_i)_{0 \leq i \leq r-1}$ of period $n$ (that is, choose $r\times n$ symbols).
	\item For each $r+1 \leq i \leq p$, choose  a
	  bi-infinite row $c_i$ of period $n$ (that is, choose $n$ symbols),
	  and verify that all $r \times r$ blocks in
	  the rows $c_i \dots c_{i-r+1}$ do not contain a forbidden pattern. At each time, keep
	  only the last $r$ rows in memory (we never forget the $r$ first rows 
	  though).
	\item (Verification of the \emph{least} period) 
	If at any of the previous steps, the row $c_i$ is not periodic of period
        $k < n$, then $P[k] \leftarrow 0$
	\item For $i \leq r$, verify that all $r \times r$ blocks in 
	  the rows
	  $c_{p-i} \dots c_p c_0 \dots c_{i-r-1}$ do not contain a forbidden pattern.
	\item If there is some $k$ such that $P[k] = 1$, reject. Otherwise accept. 
\end{itemize}
This algorithm needs to keep in memory only $2r$ rows and the array $P$ at each time, 
hence is in space $\mathcal O(n)$.
\end{proof}
\begin{lemma}\label{lem:hori:reciproque}
Let $L\subseteq \NN$ be a language such that $\unaire{L}\in\nspace{n}$, then
there exists a two-dimensional SFT $X$ such that $n\in L$ if and only if there
exists a
point $c\in X$ with horizontal period $n$.
\end{lemma}
\begin{proof}
Take a nondeterministic Turing machine $M$ accepting $\unaire{L}$ in linear space. 
Using traditional tricks from complexity theory, we can suppose that 
on input $1^n$ the Turing machine uses exactly $n+1$ cells of the tape
(i.e. the input, with one additional cell on the right) and works in time
exactly $c^n$ for some constant $c$ (depending only on the Turing machine $M$).

We are going to construct an SFT $X'$ such that $1^{n} \in L$
if and only if $n+4$ is a horizontal period of some point of $X'$.
The modification to obtain $n+1$ rather than $n+4$, and thus prove the
lemma, is left to the reader (basically ``fatten'' the vertical lines of \breaker
presented in lemma~\ref{lem:hori:structure} below so that they absorb 3 adjacent tiles), and serves no interest other than technical.

The proof may basically be split into two parts: 
\begin{itemize}
 \item First produce an SFT $Y_c$ such that
every point of horizontal period $n$ looks like a grid of rectangles
of size $n$ by $c^n$ delimited by horizontal and vertical markers (see fig.~\ref{fig:hori:shape}b) and whose horizontal periods are $\NN\setminus\{0,1\}$. 
Lemma~\ref{lem:hori:structure} shows how to construct such an SFT.

\item The Turing machine $M$ is then encoded inside these rectangles on a layer $C$: the nondeterministic transitions are \emph{synchronized} on every line to ensure the computations inside the rectangles are the same.
\end{itemize}

The main difficulty lies in the first part, while the second part is straightforward and does not need further explanation.

Now we prove that $1^{n} \in L$ if and only if $n+4$ is a horizontal period of the the SFT $X'$.
\begin{itemize} 
\item For the  component $C$ to be valid, the input $1^{p-4}$
(4 = 1 (\breaker symbol of $Y_c$) + 1 (left border for the TM) + 1 (right border for the TM) + 1 (blank marker on the end of the tape)) must be accepted by the Turing machine, 
hence $1^{p-4} \in L$
\item  Finally, due to the synchronization of the nondeterministic transitions, the
$C$ component is also $p$-periodic. As a consequence, our tiling is
$p$-periodic, hence $n = p-4$. Therefore $1^{n+4} \in L$
\end{itemize}

Conversely, suppose $1^{n} \in L$.
Consider the coloring of period $n+4$ obtained as follows (only a period is
described):
\begin{itemize}
	\item The component $A$ consists of $n+3$ correctly tiled columns of our
	  aperiodic East-deterministic SFT, with an additional column of \breaker.
	  Note that aperiodic points exist.
	\item The component $C$ corresponds to a successful computation path of
	  the Turing machine on the input $1^n$, that exists by hypothesis.
	  As the computation lasts less than $c^n$ steps, the computation
	  fits exactly inside the $n \times c^n$
	  rectangle.
	\item We then add all other layers according to the rules to obtain a
	  valid configuration, thus obtaining a point of $X'$ of period exactly $n+4$.
  \end{itemize}
\end{proof}

\begin{lemma}\label{lem:hori:structure}
There exists an SFT $Y_k$ such that for any point $y\in Y_k$ of horizontal period $p$, $y$ is composed of rectangles of size $p\times k^{p-1}$ with marked boundaries. Furthermore, all integers $n\geq 2$ are horizontal periods.
\end{lemma}
\begin{proof}

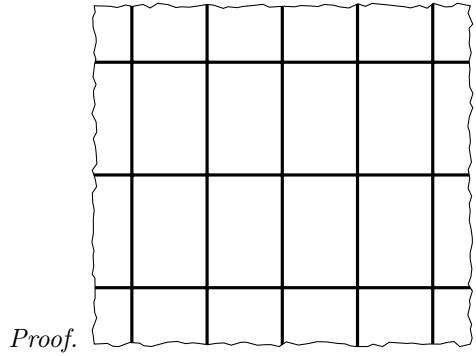
\begin{figure}[htbp]
\begin{center}
\begin{tikzpicture}[yscale=1.5]
   \clip[draw,decorate,decoration={random steps,segment length=3pt,amplitude=1pt}] 
  (0.5,0.5) rectangle (5.5,3.5);
   \draw[very thick] (0,0) grid (6,4);
   \foreach \i in {1,2,3,4,5,6}
    \foreach \j in {1,2,3,4}{
    }
\end{tikzpicture}
\end{center}
\caption{The shape of the base SFT $Y_k$: whenever a point of $Y_k$ is periodic, it must have the 
above shape where the width of the rectangles is exactly the period $p$ and their height $k^{p-1}$.}
\label{fig:hori:shape}
\end{figure}

We will construct the SFT by superimposing several components (or layers) each of them
addressing a specific issue $Y_k=A\times C_k\times T$:
\begin{itemize}
 \item $A$ will allow us to force periodic tilings to have columns separated by vertical lines,
 \item $C_k$ will make the horizontal lines and and at the same time force the regularity of the
vertical ones,
\item $T$ will force that within a horizontal period, only one vertical line can appear.
\end{itemize}

The components and their rules are as follows:
\begin{figure}[t]
 \begin{center}
 \scalebox{0.5}{
 \includepicture{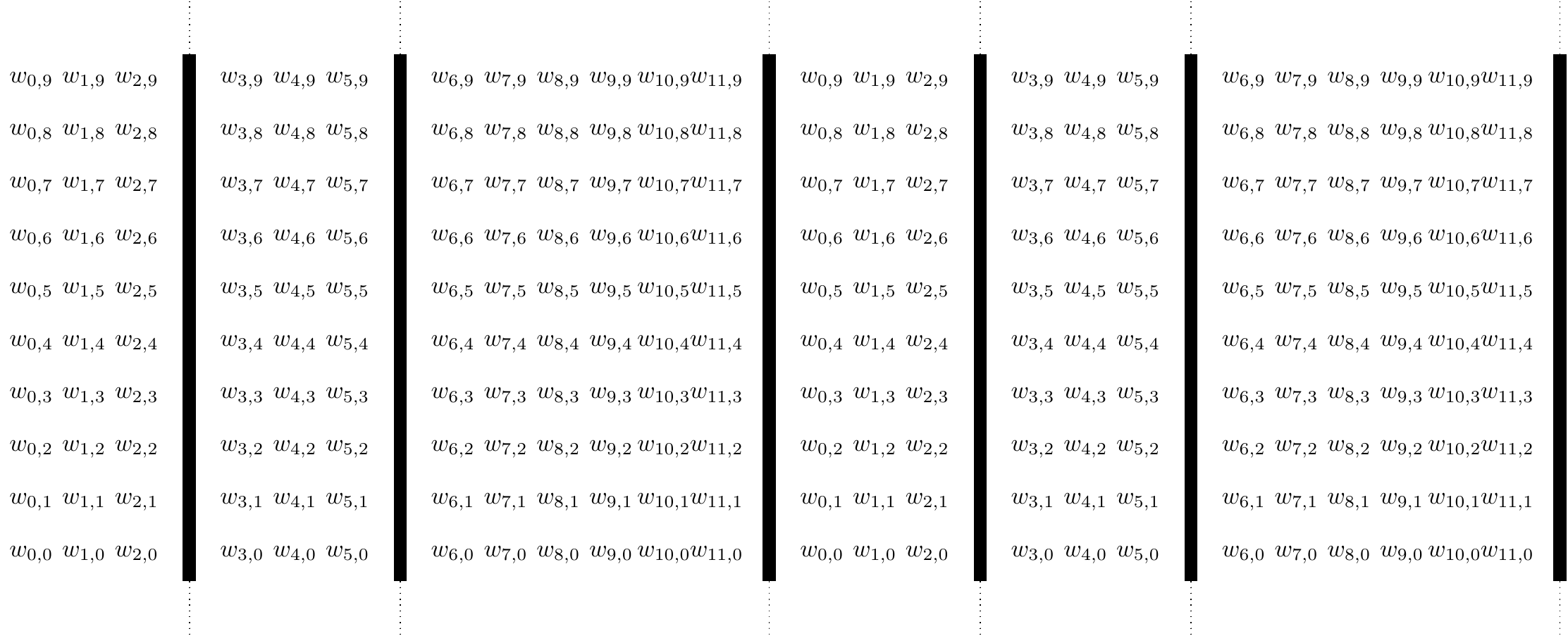}
 }
 \caption{A periodic point of $A$. Here $w_{i,j}$ are symbols of the alphabet of $W$.}\label{fig:compoA}
 \end{center}
\end{figure}

\begin{itemize}
 \item The first component $A$ is partly composed of an aperiodic East-deterministic SFT $W$,
 whose symbols will be called \emph{white} symbols. We can take the one from 
 J. Kari~\cite{KariNil}\footnote{Note that in the paper the SFT is described by Wang tilings and 
 that it is NW-deterministic. However, it is straightforward to modify the rules in order to get an 
 East-deterministic aperiodic SFT. This exact SFT will be studied later on in section~\ref{S:counting}}.
 
 To obtain $A$ from $W$, we add to the alphabet a new symbol \breaker.
 With the additional forbidden patterns:
 \begin{itemize}
  \item no white symbol shall be above or below a \breaker,
  \item two \breaker cannot appear next to each other horizontally.
 \end{itemize}
 
 With this construction, a periodic point of $A$ of period $p$ must have columns of white symbols
 separated by vertical lines of \breaker. This is due to the fact that $W$ forms an aperiodic SFT.
 
 For the moment nothing forbids more than one gray column to appear inside a period. Figure~\ref{fig:compoA} shows a possible form of a periodic point at this stage.

 \item The second layer of symbols $C_k=P_k\times\{\horline,\hempty\}$ will 
   produce horizontal lines so that points of period $n p$ will consist of rectangles of size
   $p \times k^{p}$, delimited by the symbols
   $\horline$ and $\breaker$.

 The idea is as follows: suppose each horizontal segment between two vertical lines is a
 word over the alphabet $P'_k=\left\{\{0,\dots k-1\}\times\{0,1\}\right\}$, that is, represents a number $a$ between  $0$ and $k^{p-1}-1$ written in base $k$. It is then easy with local constraints to ensure that the word on the \emph{next} line is $a+1 \mod k^{p-1}$. The $\{0,1\}$ component represents the carry. The alphabet $P_k$ is composed of $P'_k$ and of a carry ${}^1$ that will be superimposed to the symbol \breaker only.
See figure~\ref{fig:transdus} for a transducer in the case $k=2$ and its realization as an 
SFT.
 
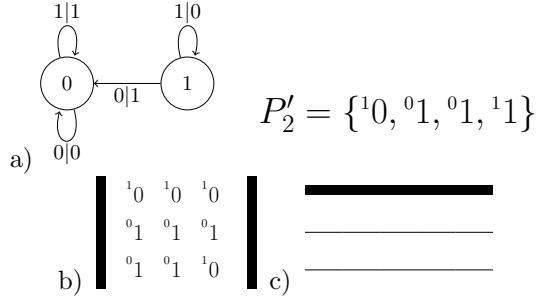
\begin{figure}[t]
 \begin{center}
  a)\scalebox{0.8}{
  \begin{tikzpicture}[auto,inner sep=1]
 \node[state] 	(q_0)		{$0$};
 \node[state]	(q_1)   at (2,0)	{$1$};
 \path[->]	(q_0)	edge [loop above]	node 	{$1|1$}	(q_0)
 		(q_0)	edge [loop below]	node	{$0|0$}	(q_0)
		(q_1)	edge 	node	{$0|1$}	(q_0)
		(q_1)	edge [loop above]	node	{$1|0$}	(q_1);
\begin{scope}[shift={(5,-1)},scale=0.25]
 \node (valeur) at (2,2) {\huge $P'_2=\{{}^{\textrm{\small 1}}0,{}^{\textrm{\small 0}}1,{}^{\textrm{\small 0}}1,{}^{\textrm{\small 1}}1\}$};
\end{scope}
\end{tikzpicture}
}
b)~
\scalebox{0.5}{
  \begin{tikzpicture}[scale=0.25]
\begin{scope}
 \node (valeur) at (2,2) {\huge 1};
 \node (retenue) at (1,3) {\small 0};
\begin{scope}[shift={(4,0)}]
 \node (valeur) at (2,2) {\huge 1};
 \node (retenue) at (1,3) {\small 0};
\end{scope}
\begin{scope}[shift={(8,0)}]
 \node (valeur) at (2,2) {\huge 0};
 \node (retenue) at (1,3) {\small 1};
\end{scope}
\begin{scope}[shift={(0,4)}]
 \node (valeur) at (2,2) {\huge 1};
 \node (retenue) at (1,3) {\small 0};
\end{scope}
\begin{scope}[shift={(4,4)}]
 \node (valeur) at (2,2) {\huge 1};
 \node (retenue) at (1,3) {\small 0};
\end{scope}
\begin{scope}[shift={(8,4)}]
 \node (valeur) at (2,2) {\huge 1};
 \node (retenue) at (1,3) {\small 0};
\end{scope}
\begin{scope}[shift={(0,8)}]
 \node (valeur) at (2,2) {\huge 0};
 \node (retenue) at (1,3) {\small 1};
\end{scope}
\begin{scope}[shift={(4,8)}]
 \node (valeur) at (2,2) {\huge 0};
 \node (retenue) at (1,3) {\small 1};
\end{scope}
\begin{scope}[shift={(8,8)}]
 \node (valeur) at (2,2) {\huge 0};
 \node (retenue) at (1,3) {\small 1};
\end{scope}
\begin{scope}[shift={(-4,8)},scale=4]
\sbreaker
\end{scope}
\begin{scope}[shift={(-4,4)},scale=4]
\sbreaker
\end{scope}
\begin{scope}[shift={(-4,0)},scale=4]
\sbreaker
\end{scope}
\begin{scope}[shift={(12,0)},scale=4]
\sbreaker
\end{scope}
\begin{scope}[shift={(12,4)},scale=4]
\sbreaker
\end{scope}
\begin{scope}[shift={(12,8)},scale=4]
\sbreaker
\end{scope}
\end{scope}
\end{tikzpicture}
}~~~c)~
\scalebox{0.5}{
\begin{tikzpicture}[scale=0.25]
\fill[color=white] (0,0) rectangle (1,1);
\begin{scope}[shift={(-2,2)}]
\begin{scope}[scale=4]
 \foreach \i in {0,...,4}{
  \begin{scope}[shift={(\i,2)}]
   \shorline
  \end{scope}
  \foreach \j in {0,...,1}{
   \begin{scope}[shift={(\i,\j)}]
    \shempty
   \end{scope}
 }
}
\end{scope}
\end{scope}
\end{tikzpicture}
}
 \end{center}
 \caption{a) The transducer corresponding to $k=2$ and the symbols of the corresponding SFT. A valid pattern for $C_2$ is given in b) and c): the breaker symbol in b) is from the $A$ layer and corresponds always to the ${}^1$ symbol of $C_2$.}
 \label{fig:transdus}
\end{figure}
 

With the $\{\horline,\hempty\}$ subcomponent, we mark the lines corresponding to the number $0$, so that one line out of $k^{p-1}$ is marked. This line is the only one where a ${}^01$ is transformed into a ${}^10$ on the right of a vertical line of \breaker.

A \hempty is forbidden to appear on the right or left of a \horline: this forces each column to have counters that are resetted at the exact same moment, and thus to have the exact same size. Figure~\ref{fig:compoDT}a shows some typical tiling at this stage:
 the period of a tiling is not necessarily the same as the distance between
 the rectangles, it may be larger. Indeed, the white symbols in two consecutive rectangles
 may be different.

 \item Component $T$ is formed of the same alphabet as $W$ of component $A$, recall that $W$ is an East-deterministic SFT. The forbidden patterns are that two different symbols cannot be horizontal neighbors. In addition to that, we forbid of elements of component $T$ on the right of a \breaker to be different to the ones of component $W$. That means that the symbols of the first column to the right of a vertical line of \breaker are exactly the same for each vertical line of \breaker.
 
 The SFT $W$ being East-deterministic, this means that the symbols of the
columns between vertical lines of \breaker are exactly the same for each column
, these are the boundaries of the rectangles.
 
 At this stage, a periodic point necessarily has regular rectangles on all the plane, whose width correspond to the period, as shown in figure~\ref{fig:compoDT}b. 
\begin{figure}[tbp]
 \begin{center}
\scalebox{0.5}{
\includepicture{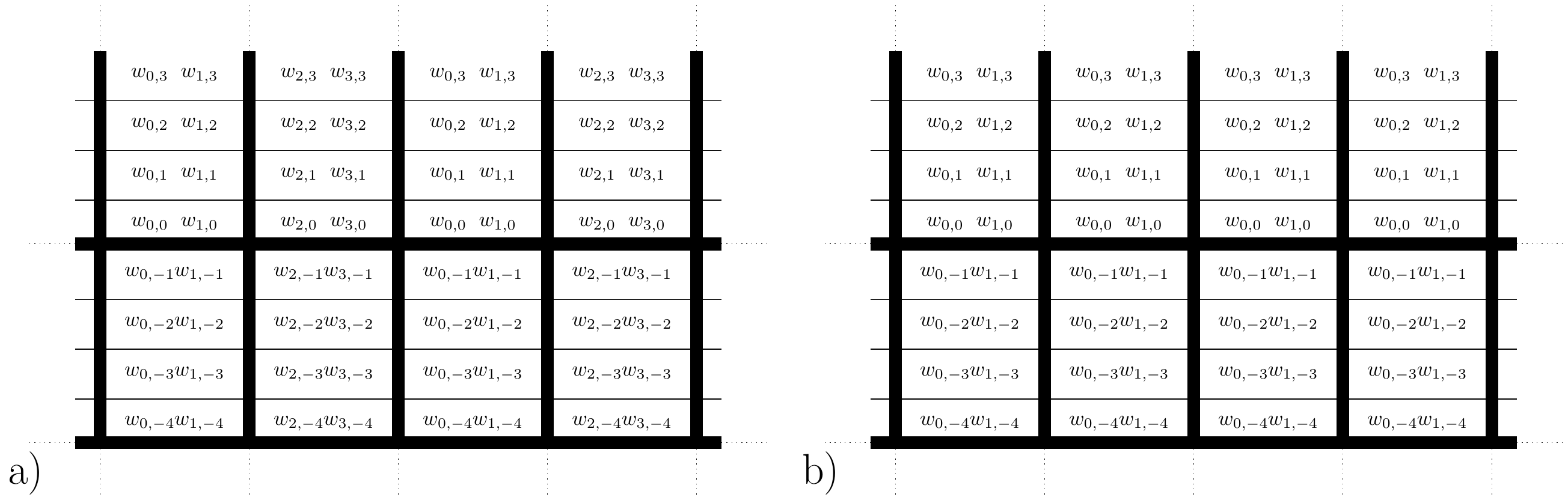}
 }
 \caption{ a) An example of a valid periodic point with $C_2$, the distance between two consecutive vertical lines  is now constant. However the period is not necessarily the width of the rectangles. b) Once we add component $T$, the width of the rectangles is now exactly the period.
 Note that we only show here the component $A$ and the subcomponent
 $\{\protect\horline,\protect\hempty\}$ of $C_2$.
 }\label{fig:compoDT}
 \end{center}
\end{figure}
\end{itemize}

Now we prove that if $y\in Y_k$ is periodic of period $n$, then it necessarily is formed of 
vertical lines of $\breaker$ at distance $n$ of each other and of horizontal lines of $\horline$
at distance $k^{n-1}$ of each other.

Consider a point of $Y_k$ of period $n$:
\begin{itemize}
\item Due to component $A$, a vertical line of \breaker must appear. The period is a
succession of vertical lines of \breaker and white columns.
\item  Due to component $C_k$, the vertical lines of \breaker are spaced by a distance of $p$, where $p \mid n$. Furthermore, there are horizontal lines of \horline at distance $k^{p-1}$ of each other.
\item  Due to component $T$, the tiling we obtain is horizontally periodic of period $p$, thus
$p=n$.
\end{itemize}

\end{proof}

\newpage
\section{Strong periodicity in SFTs and time complexity}

In this section we prove theorem~\ref{thm:strong} on strong periods.

In a square $n\times n$ one can only embed computations ending in time inferior
to $n$. However, given a $n\times n$ square of symbols, one
cannot check that it has no forbidden patterns in less than $n^2$ time
steps with a Turing machine. 

The analogue for higher dimensions holds: any $2d$-dimensional cube of $n^{2d}$ symbols
can embed computations in time $n^d$, and checking such a cube needs $n^{2d}$
time steps. Thus the class $\ntimeun{n^d}$ for unary inputs\footnote{
Note that the class $\ntimeun{n^d}$ for unary inputs corresponds to the class
$\ntimebin{2^{dn}}$ for binary inputs and that complexity classes are usually
defined on binary inputs. So $\NP_1$  is not the famous $\NP$ class, although $\NP=\P$ 
would imply $\NP=\P$.} 

can be captured by $2d$-dimensional cubes and only $d$-dimensional cubes are
checkable in time $n^d$

The gap here is not surprising: while space complexity classes are usually
model independent, this is not the case for time complexity, where the exact
definition of the computational model matters. An exact characterization of
strongly periodic SFTs for $d=2$ would in fact be possible, but messy: it would
involve Turing machines working in space $O(n)$ with $O(n)$ reversals, see e.g
\cite{CaraMey}.

Here the solution comes from the fact that periodic SFTs of all dimensions would
be captured by the time complexity class $\NP_1=\bigcup_{d\in\NN}\ntimeun{n^{d}}$,
the gap being filled by the infinite union. This is theorem~\ref{thm:strong},
whose proof below will be, as before, divided in three parts :
\begin{itemize}
 \item We first show that one can check whether a $d$ dimensional SFT is strongly 
 periodic of period $p$ in time $p^d$ and thus that the problem is in 
 $\ntimeun{n^{d}}$ (lemma~\ref{lem:strong:direct}).
 \item For the converse,  we first construct a base SFT with marked cubes 
 (lemma~\ref{lem:strong:struct}) in a similar way as for horizontal periods,
 \item All that is left to prove then is how the Turing machines are encoded 
 inside these cubes  (lemma~\ref{lem:strong:reciproque}). 
\end{itemize}

It is interesting to note that the complexity class $\NP_1$ also characterizes spectra 
of first order formula,  see \cite{JonesSelman}.

\begin{lemma}\label{lem:strong:direct}
 For any $d$ dimensional SFT $X$, $\unaire{\strongper{X}}\in\NP$:
 there exists a Turing machine $M\in\NP_1$ that given $p$ as an input, 
 determines whether $p$ is a strong period of $X$. 
\end{lemma}
\begin{proof}
 It suffices to take a Turing machine that nondeterministically guesses a $p^d$ cube 
 and then checks whether it contains any forbidden patterns and whether $p$ is the strong 
 period of $X$ : for this last part, it has to check that for all $k<n$, $k$ is not a period.
\end{proof}

\begin{lemma}\label{lem:strong:struct}
There exists a $d$-dimensional SFT $Y_d$ such that :
\begin{itemize}
 \item Any periodic point is strongly periodic.
 \item Any strongly periodic point $y\in Y_d$ of period $p$ 
 is constituted of adjacent $d$-cubes $p^d$ with marked borders.
 \item Every integer $p\geq 2$ is a strong period of $Y_d$.
\end{itemize}
\end{lemma}

\begin{proof}
As before, the construction will be based on some aperiodic SFT, with some added symbols 
to break the aperiodicity and force a regular structure. The SFT $Y_d$ is made of three layers $A\times S\times T$:
\begin{itemize}
 \item layer $A$ will force strongly periodic points to have marked lines,
 \item layer $S$ will force strongly periodic points to have marked $d$-cubes,
 \item and finally, layer $T$ will force the strongly periodic points to be composed of $d$-cubes 
 of side $p$, the strong period.
\end{itemize}

We will now detail each component and the strongly periodic points thus obtained.
\begin{itemize}
\item Let $W$ be a 2-dimensional NW-deterministic aperiodic SFT\footnote{For this lemma we can take any such SFT,
however we will use Kari's~\cite{KariNil} tileset later to further the construction.
}, 
we define a 2-dimensional SFT $A'$:
\begin{itemize}
 \item the alphabet of $A'$ is composed of the alphabet of $W$ with the additional symbols 
 $\breaker,\corner,\horline$,
 \item the forbidden patterns of $W$ are kept and the following are added:
\begin{itemize}
 \item above and bellow a \breaker may only appear a \breaker or a \corner,
 \item on the left and right of a \horline may only appear a \horline or a \corner,
 \item on the left and right of a \corner there may only be a \horline,
 \item above and below a \corner, there may only be a \breaker.
\end{itemize}
\end{itemize}
The (strongly) periodic points of $A'$ have necessarily one of the new symbols 
$\breaker,\corner,\horline$. That is to say they are necessarily formed of either an infinity of lines 
of \breaker, either an infinity of lines of \horline, or by an infinity of
squares with sides marked by \horline and \breaker and corners by \corner.

The $d$-dimensional SFT $A$ is obtained from the 2-dimensional SFT $A'$ by keeping the alphabet,
keeping the rules for the first two dimensions, and then force the symbols next to each 
other along all other directions to be equal. In the sequel, we will call $A'$ the plane with the
rules of $A'$.

\item The second layer $S$ will force the $A'$ plane of periodic points to be formed of 
squares and will mark the frontiers of the $d$-cubes. 
For $2\leq i\leq d$, we define $d-1$ SFTs $S_i$ with the same alphabet formed by the symbols 
$\diagdiag,\diagver,\diaghor,\diagcorner,\diagleft,\diagright$:
\begin{itemize}
 \item The adjacency rules for $S_i$ 
on the plane defined by $e_1, e_i$ are that two symbols can be 
next to each other iff their borders match: left/right matchings correspond to $\pm e_i$ 
and above/below to $\pm e_1$. 
 \item The rules applying to the other dimensions are
that if there is a symbol $a$ at $x\in\ZZ^d$, then there must also 
be a symbol $a$ at $x\pm e_k$, for $k\neq 1,i$. 
\end{itemize}

We superimpose the $S_i$'s in order to obtain $S$: at each position, the symbols on all $S_i$ components must all be 
taken from only one of the sets $\{\diagcorner,\diaghor\}$ and $\{\diagver,\diagleft,\diagright,\diagdiag\}$. 

See figure~\ref{fig:strong:cubes3d} for an example of how $S_i$ and $S_j$ can be
superimposed.

\begin{figure}
\centering
\scalebox{0.9}{\includepicture{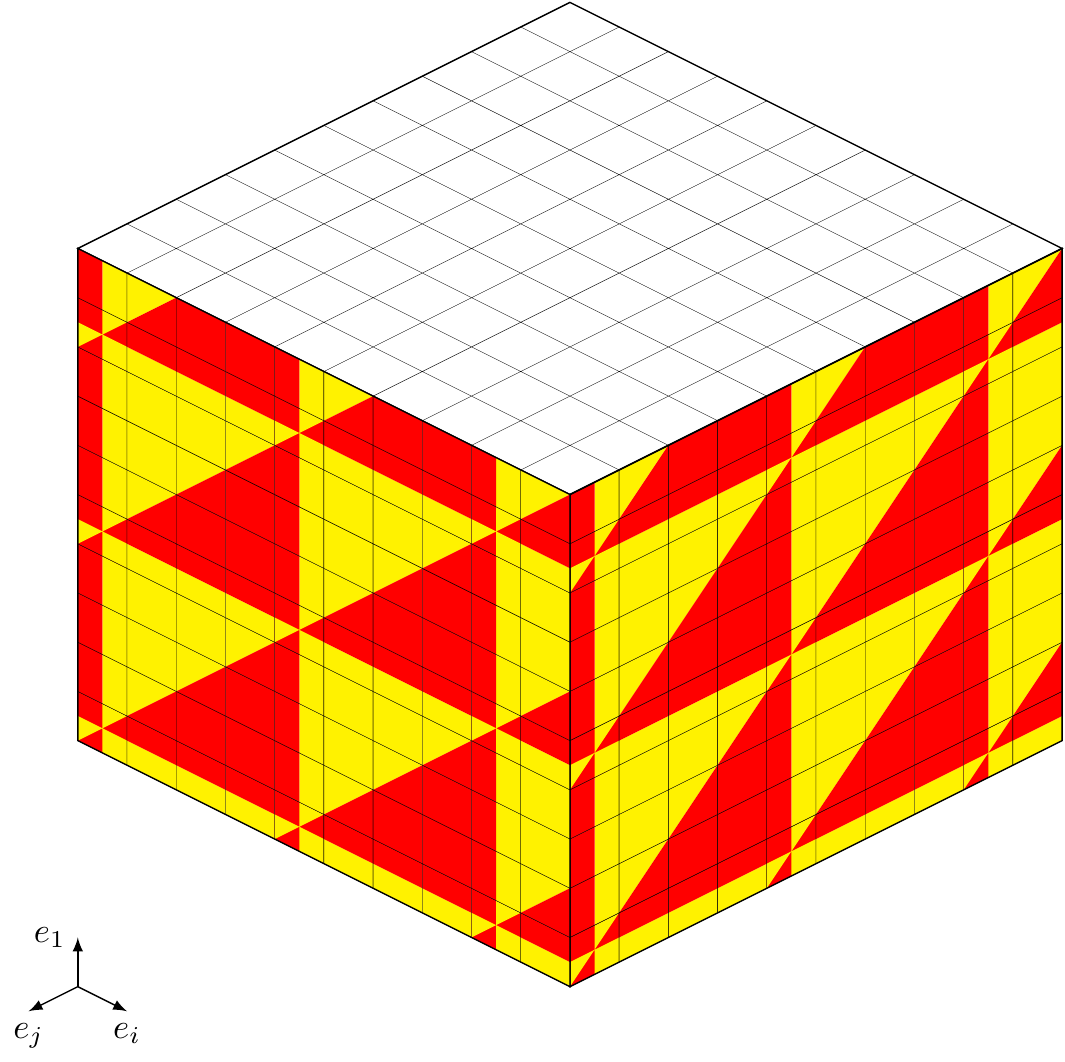}}
 \caption{\label{fig:strong:cubes3d} How $S_i$ and $S_j$ are superimposed.}
 
\end{figure}

Then to obtain $A\times S$ we add superimposition rules only with the 
$S_2$ subcomponent of $S$,  which has its rules on the $A'$ plane and forms squares on it. 
The rules are that $\diagcorner$ can only be superimposed to 
$\corner$, \diaghor (resp. \diagver) can only be superimposed to \horline (resp. \breaker) and the
other symbols can only be superimposed to white symbols. As a consequence, the symbols \breaker,\horline,\corner on
the $A'$ plane must form squares on the strongly periodic points. Figure~\ref{fig:strong:superimposition} 
shows how the $A'$ planes of component $S_2$ and component $A$ are superimposed.

\begin{figure}
 \centering
 \begin{center}
  \includepicture{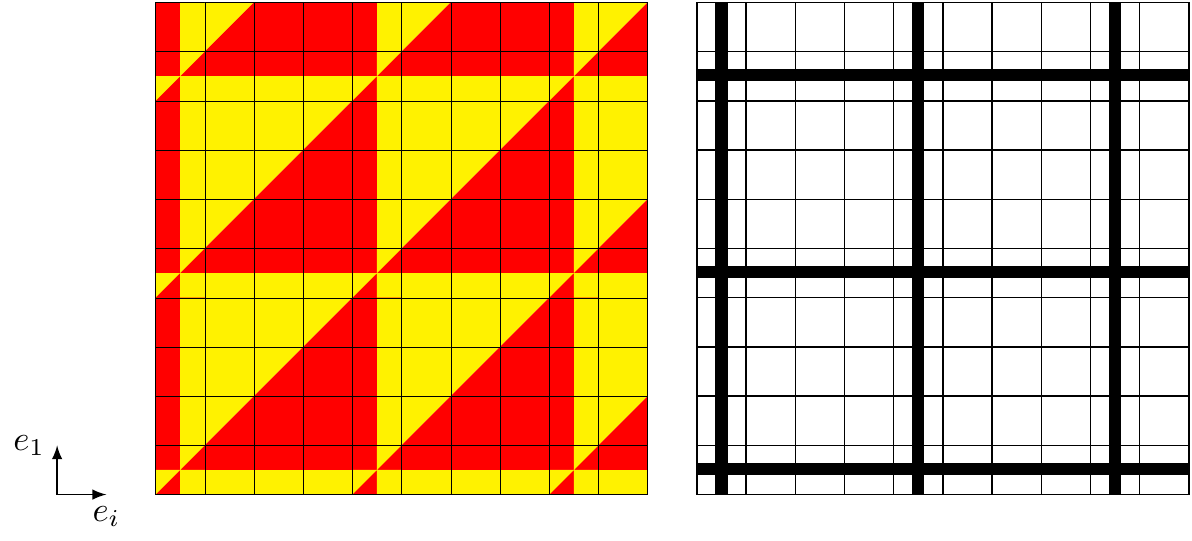}
 \end{center}
 \caption{\label{fig:strong:superimposition}The $A'$ planes of components $S_2$ (left) and $A$ (right) and how they are superimposed: 
 layer $S_2$ forces the symbols \protect\breaker,\protect\horline and \protect\corner of component $A$ to form squares.}
\end{figure}

The strongly periodic points of the resulting SFT are points that have 
$d$-cubes whose corners are marked by $(\diagcorner, \dots, \diagcorner)$. The boundaries of the 
$d$-cubes are marked by the  \diaghor and 
\diagver symbols: if the side of the $d$-cubes is $n$, and there is a corner
at $\vect p=(p_1,\dots,p_d)\in\ZZ^d$, then for any point $\vect q=(q_1,\dots,q_d)\in\ZZ^d$, a
\diagver or \diagcorner on component $S_k$ is equivalent to $p_k\equiv q_k \mod
n$ and a \diaghor or \diagcorner is equivalent to $p_1\equiv q_1 \mod n$. 

However, the squares formed on the $A'$ plane of $A$ may not have the same aperiodic background, and 
thus there could be more than one square in a period. Thus the period is a multiple of the size of the $d$-cubes. 
We want now to prevent this from happening and force the strong period to be exactly the size of the $d$-cubes.

\item Component $T=L_{right}\times L_{diag} \times U_{up}\times U_{diag}$ is
here to address this last problem: by synchronizing the aperiodic 
background between squares on $A'$ it will force the least distance between two 
$(\diagcorner,\dots,\diagcorner)$ to be the strong period. To do this, since $W$ is NW-deterministic, we only need to
transmit to the neighboring squares the upper line of symbols and the leftmost one, see 
figure~\ref{fig:strong:transmiNW}. 
Each subcomponent's alphabet is a copy of the alphabet of $W$ and the rules
are as follows:
\begin{itemize}
 \item On $L_{right}$, the symbols at $\vect z$ are the same as the symbols at
$\vect z\pm \vect{e_2}$. The only superimposition rule is that a symbol on
the right of a \breaker on $A'$ must be the same as on $L_{right}$. This component synchronises 
the leftmost column of symbols of all horizontally aligned squares.
 \item On $L_{diag}$, the symbols at $\vect{z}$ are the same as the symbols at
$\vect z\pm (\vect{e_1}+\vect{e_2})$. The only superimposition rule is that a
symbol on the right of a \breaker on $A'$ must be the same as on $L_{diag}$.
This component synchronises the leftmost column of symbols of all diagonally aligned squares.
 \item On $U_{up}$, the symbols at $\vect z$ are the same as the symbols at
$\vect z\pm \vect{e_1}$. The only superimposition rule is that a symbol above a  \horline
on $A'$ must be the same as on $U_{up}$. 
 \item On $U_{diag}$, the symbols at $\vect z$ are the same as the symbols at
$\vect z\pm (\vect{e_1}+\vect{e_2})$. The only superimposition rule is that a symbol a \horline on
$A'$ must be the same as on $U_{diag}$. 
\end{itemize}

\begin{figure}
 \centering
 \scalebox{0.5}{\includepicture{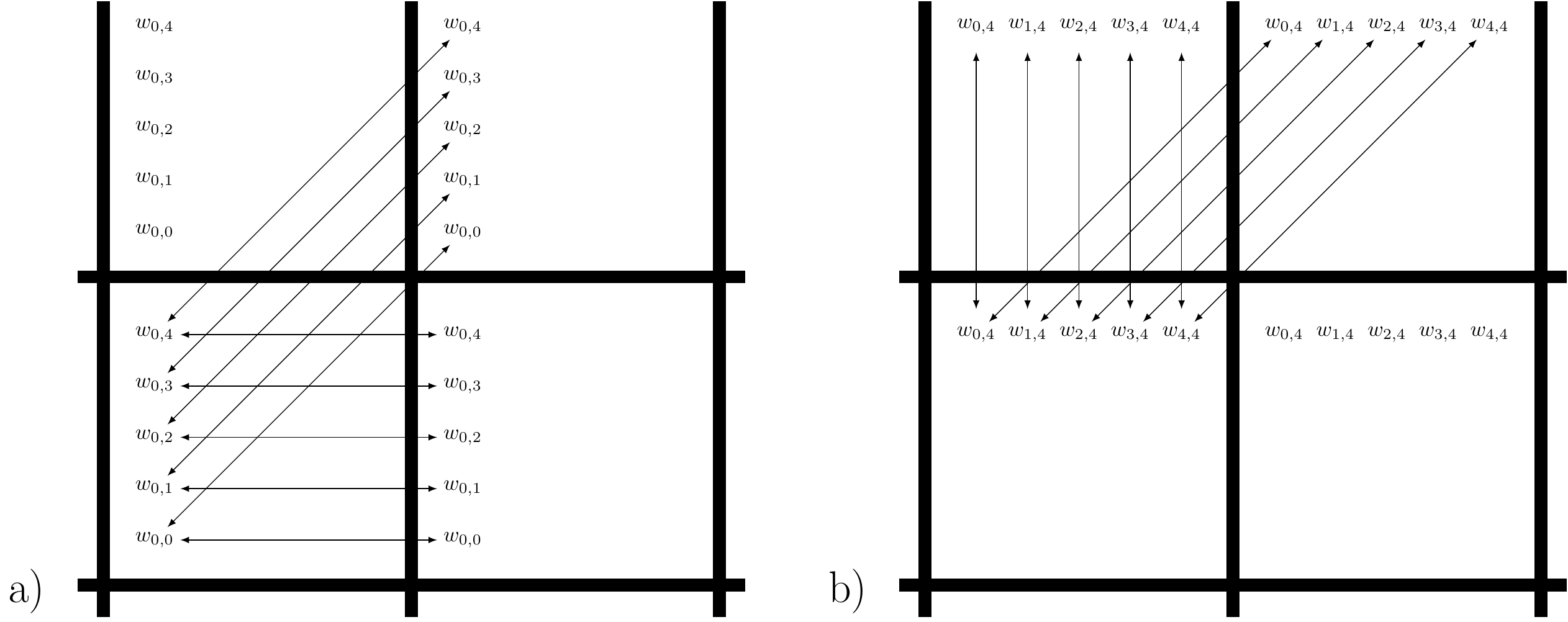}}
  \caption{\label{fig:strong:transmiNW}
  In a), the way $L_{right}$ and $L_{diag}$ synchronise the first column of the
aperiodic background for all squares. In b), the way $U_{up}$ and $U_{diag}$
synchronise the top line of the aperiodic background for all squares.}
\end{figure}
\end{itemize}

The construction is now finished. Now take a strongly periodic point
of $Y_d$ with strong period $p$. By construction, it necessarily is
constituted of identical adjacent hypercubes of side $p$. The boundaries of the
hypercubes being marked by symbols $(\diaghor,\dots,\diaghor)$ and
$(\diagver,\dots,\diagver)$ and the corners by
$(\diagcorner,\dots,\diagcorner)$. It is straightforward to see that any $p\geq 2$ is a strong 
period of the constructed subshift.
\end{proof}

\begin{lemma}\label{lem:strong:reciproque}
 Let $L\subseteq \NN$ be a language such that $\unaire{L}\in\ntime{n^d}$,
there exists a $2d$-dimensional SFT $X$ such that $L=\strongper{X}$ and such that all 
periodic points are strongly periodic.
\end{lemma}

\begin{proof}
Let $M$ be a Turing machine recognizing $L$ in nondeterministic time $n^d$. We need to 
construct an SFT $X_M$ whose strong periods are exactly the accepted inputs of $M$. We
Using lemma~\ref{lem:strong:struct}, all that is left to prove is how to restrict the 
periods to the integers accepted by $M$. In order to do this,
we will encode computations of $M$ inside the $2d$-cubes of $Y_{2d}$: on a
unary input, $M$ takes at most $n^d$ time steps to accept or reject, so a $2d$-cube of side 
$n$ has exactly the right amount of space to encode such a computation. 
The idea is to \emph{fold} the space-time diagram of the Turing machine so that it fits into 
the cube while still preserving the local constraints. Such a folding has already been described by Borchert
\cite{Borchert} and can also be deduced from Jones/Selman \cite{JonesSelman}.
We then have to make sure that the nondeterministic transitions are identical in all $2d$-cubes of a point.
Let us now describe this in more details.

In a space-time diagram of $M$ with input $n$, tape cells have coordinates 
$(t,s)$ with $t \leq n^d, s\leq n^d$, where $t$ is the time step and $s$ the position in space.
We now have to transform each cell $(t,s)$ into a cell of the $2d$-cube of size $n$,
so that two consecutive (in time or space) cells of the space-time diagram remain
adjacent cells of the cube. So we transform $s$ and $t$ 
into elements of $\llbracket 0,n-1\rrbracket^d$ with a reflected $n$-ary code 
(also called reflected Gray-codes), see \cite{Flores,Knuth4}, this corresponds exactly 
to folding the time/space.

The vector $(t_0, \dots, t_{d-1}) \in \llbracket 0,n-1\rrbracket^d$
will represent the integer $t = \sum a_i n^i$ where 
\[
 a_i = \left\{
  \begin{array}{cl}
   t_i & \text{when }\sum_{j>i} t_j\text{ is even} \\
   (n-1)-t_i & \text{otherwise} \\
  \end{array}
  \right.
\]
see \cite[Formula (51)]{Knuth4}. The next positition is given by the
parity of the sum of the \emph{stronger weighed} digits. In order to tranform this in local
constraints, it will suffice to encode parities of positions in the cube with some layers 
$P^t$ and $P^s$ for time and space respectively.

Layer $P^t$ is made of several sublayers $P_i=\{0,1\}$, one for each direction $\vect{e_i}$, 
$2\leq i\leq d$. We now give the rules, recall that the boundaries of the cube are 
marked. Without loss of generality, we may suppose that there is a corner in position 
$\vect 0$. This corner has $0$ on all layers $P_i$. The rules are the following : if there 
the symbol $p$ at position $\vect z\in\llbracket 0,n-1\rrbracket^d$ on sublayer $P_i$, then 
there must be $p+1 \mod 2$ at position $\vect z + \vect{e_i}$ and $p$ at positions $\vect z
+ \vect{e_j}$, with $j\neq i$. These rules do not apply when the next position is at 
the boundary of the $2d$-cube. The layer $P^s$ is similar, except it is on dimensions $d+1$ to
$2d$. An example for a three dimensional folding can be
seen on figure~\ref{fig:strong:fold3d}.
\newcommand{\even}{0}
\newcommand{\odd}{1}

\begin{figure}
  \begin{center}
    \includepicture{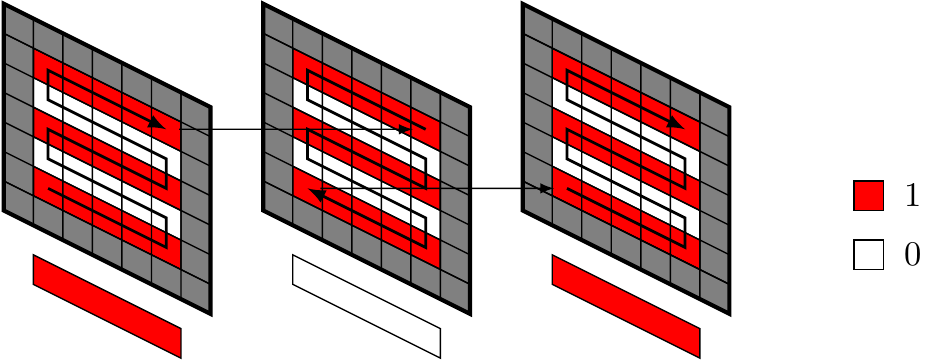}
  \end{center}
  \caption{Folding of a three dimensional cube, the red on the parity layer stands for $\even$ and
the white for $\odd$. The direction where to look for the next cell is given by the sum of the parities.}
  \label{fig:strong:fold3d}

\end{figure}

Now that we have encoded the Turing machines inside the $2d$-cubes, their size can only be
one of its accepted inputs. However, recall that the Turing machines encoded are 
nondeterministic, therefore we have to  synchronize the transitions between the different 
hypercubes, otherwise the periods may be multiples of accepted inputs. In order to do that, we
add a new component $N$, which is constituted of the following sublayers, whose alphabets are 
each a copy of the possible transitions of $M$:
\begin{itemize}
 \item The first sublayer, $\daleth$, will propagate the transition of a time-step to all 
 cells of the same time-step, that is to say the rest of the tape. A cell where a
 transition happens imposes the symbol on $\daleth$ to be the transition happening. The symbol
 propagates along space: if there is a symbol $l$ on $\daleth$ at position $\vect z\in\ZZ^d$, 
 then there must be exactly the same symbol at position $\vect z\pm \vect{e_i}$, 
 for $d+1\leq i\leq 2d$.
 \item We also have a set of sublayers $\beth_i$, one for each time dimension, $1\leq i\leq d$.
 Component $\beth_i$ has the following rules : the symbol on $\beth_i$ at position $\vect z$ is 
 identical to the one at position $\vect z + (\vect{e_i} + \vect{e_{d+1}})$. When the cell is on a 
 border of the $2d$-cube on dimension $1$, the symbols on $\beth_i$ and $\daleth$ have 
 to be identical. For the construction of lemma~\ref{lem:strong:struct}, this means that $S_i$
 contains a $\diagver$ or a $\diagcorner$. Figure~\ref{fig:strong:synchroNonDet} shows how this 
 synchronisation is done. 

\begin{figure}[h]
 \centering
 \scalebox{0.5}{\includepicture{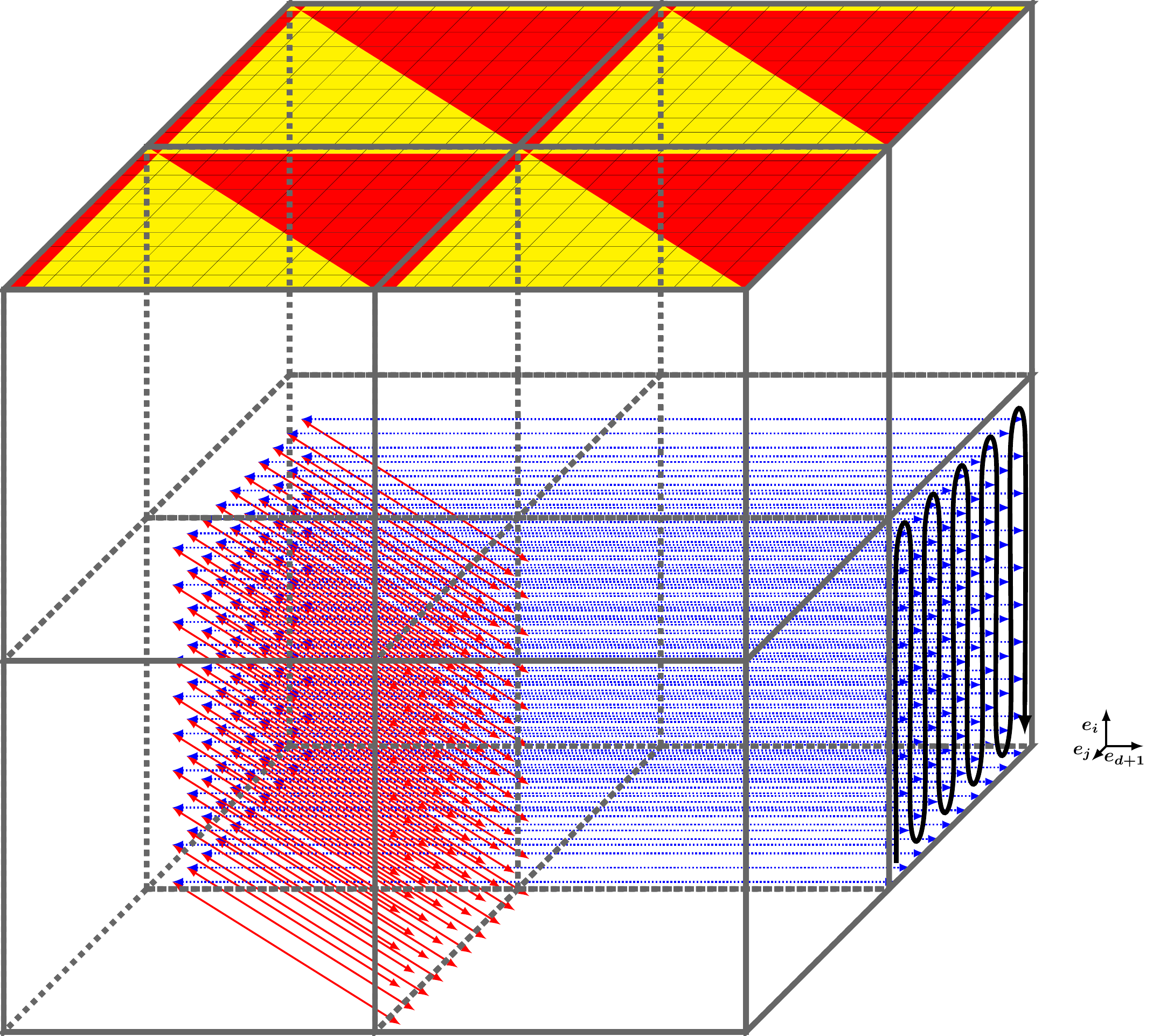}}
 \caption{\label{fig:strong:synchroNonDet} The Synchronization of nondeterministic transitions 
 between the $2d$-cubes : here is the projection in 3 dimensions  
 $\vec{e_{d+1}},\vec{e_i},\vec{e_j}$, with $1\leq i,j\leq d$. 
$\daleth$ is represented in blue and $\beth_j$ is represented 
 in red. On the top, we represented component $S_j$ of the construction of 
 lemma~\ref{lem:strong:struct}. Layers $\daleth$ et $\beth_j$ synchronise together 
 when we are on the side of a square on $S_j$.}
\end{figure}

\end{itemize}

As in lemma~\ref{lem:hori:reciproque}, $n+4$ is a (strong) period if and only if $n$ is 
accepted by $M$. Again, to obtain exactly $n$ it suffices to fatten the symbols on the borders.

\end{proof}

\section{Counting the number of periodic points in SFTs}\label{S:counting}

In theorem~\ref{thm:strong} we have seen that the sets of strong periods of SFTs
are exactly the sets of integers recognized nondeterministic polynomial time: we can go one 
step further and give a characterization of the sequences $p_n(X)_{n\in\NN^*}$ where $p_n(X)$ 
corresponds to the number of points in $X$ with period $n$. 

In the previous construction this number was related to the number of accepting paths of
the Turing machine : the number of possible aperiodic backgrounds possible for each square of $A'$
makes it hard to characterize. The following theorem is a consequence of forcing the aperiodic 
background of squares of the same size to be unique:

\begin{theorem}\label{thm:counting}
For any SFT $X$, let $\nbper{X}{}:{a}^*\to \NN$ be the function defined by $\nbper{X}{a^n}=p_n(X)/n^d$ where $d$ is
the dimension of $X$. We have then the following:
$$\left\{\nbper{X}{}\mid X\textrm{ an SFT}\right\}=\csharp$$
\end{theorem}

Note that the function $p_n(X)$ has been normalized by $n^d$, this is due to the fact that 
there are exactly $n^d$ shifted versions of a same strongly periodic point of period $n$.

\begin{proof}
To prove that $\csharp\subseteq\left\{\nbper{X}{}\mid X\text{ an SFT}\right\}$ we have to fix 
  the aperiodic background for the squares of our previous construction: the
  number of strongly periodic points of period $n$ with a corner at $\vect 0$ will 
  then be exactly the number of accepting paths of the Turing machine $M$ ran on $n$.

  In order to do that, instead of taking any NW-deterministic aperiodic SFT $W$, we 
  will take Kari's SFT \cite{KariNil} and show that it is easy to fix the top and leftmost
  borders of the squares: this will determine the rest of the square.

Kari's set of tiles is exactly the same as Robinson's \cite{Robinson}, see figure~\ref{fig:Robtiles},
except that it has one supplementary layer with diagonal arrows, see figure~\ref{fig:diagNW}. A 
valid tiling with this tileset can be seen on figure~\ref{fig:karitiling}, the top and left borders
determine the whole square. These borders are almost trivial and can be extended to any length,
still forcing an admissible pattern.

\begin{figure}[b]
 \begin{center}
  \begin{tikzpicture}[scale=0.4]
    \node at (-4,2) {$(a)$};
   \begin{scope}[shift={(-1,0)}]
    \cross
   \end{scope}
    \node at (2,2) {$(b)$};
  \begin{scope}[shift={(5,0)}]
   \armoo
  \end{scope}
  \begin{scope}[shift={(10,0)}]
    \armod
  \end{scope}
  \begin{scope}[shift={(15,0)}]
    \armdlo
  \end{scope}
  \begin{scope}[shift={(5,-5)}]
    \armdld
  \end{scope}
  \begin{scope}[shift={(10,-5)}]
    \armdro
  \end{scope}
  \begin{scope}[shift={(15,-5)}]
    \armdrd
  \end{scope}
\end{tikzpicture}
 \end{center}
 \caption{Robinson's aperiodic tileset : a) a cross and b) arms. The tileset also includes the
 rotates of these tiles. If the main arrow of an arm is horizontal (resp. vertical) we will call 
 it a  horizontal (resp. vertical) arm. Two tiles can be neighbors if and only if outgoing arrows
 match incoming ones.}
 \label{fig:Robtiles}
\end{figure}
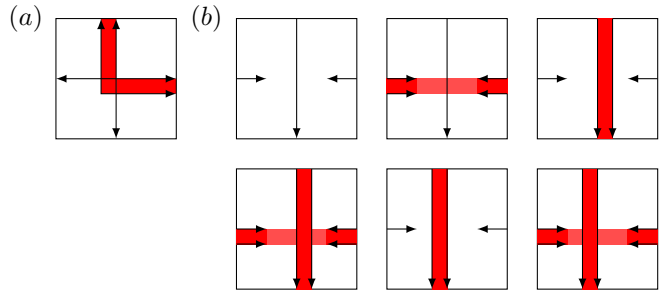

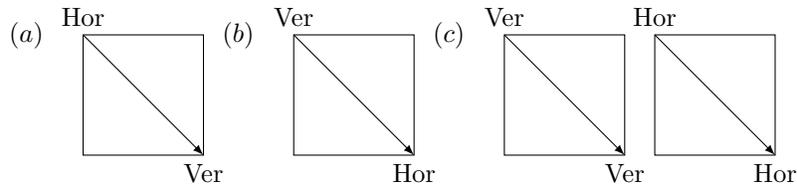
\begin{figure}[p]
 \begin{center}
  \begin{tikzpicture}[scale=0.4]
    \node[left] at (-2,4) {$(a)$};
  \begin{scope}[shift={(-1,0)}]
    \draw (0,0) rectangle +(4,4);
    \draw[-latex] (0,4) -- (4,0);
    \node[above] at (0,4) {Hor};
    \node[below] at (4,0) {Ver};
  \end{scope}
    \node[left] at (5,4) {$(b)$};
  \begin{scope}[shift={(6,0)}]
    \draw (0,0) rectangle +(4,4);
    \draw[-latex] (0,4) -- (4,0);
    \node[above] at (0,4) {Ver};
    \node[below] at (4,0) {Hor};
  \end{scope}
    \node[left] at (12,4) {$(c)$};
  \begin{scope}[shift={(13,0)}]
    \draw (0,0) rectangle +(4,4);
    \draw[-latex] (0,4) -- (4,0);
    \node[above] at (0,4) {Ver};
    \node[below] at (4,0) {Ver};
  \end{scope}
  \begin{scope}[shift={(18,0)}]
    \draw (0,0) rectangle +(4,4);
    \draw[-latex] (0,4) -- (4,0);
    \node[above] at (0,4) {Hor};
    \node[below] at (4,0) {Hor};
  \end{scope}
\end{tikzpicture}
 \end{center}
 \caption{Kari's addition to make it NW-deterministic : a new layer with diagonal arrows that
 have to match at their extremities. On horizontal arms only, we superimpose the tile (a) and on
 vertical ones only, the tile (b). The tiles (c) are superimposed on crosses only.}
 \label{fig:diagNW}
\end{figure}

\begin{figure}[b]
  \begin{center}
    \includepicture{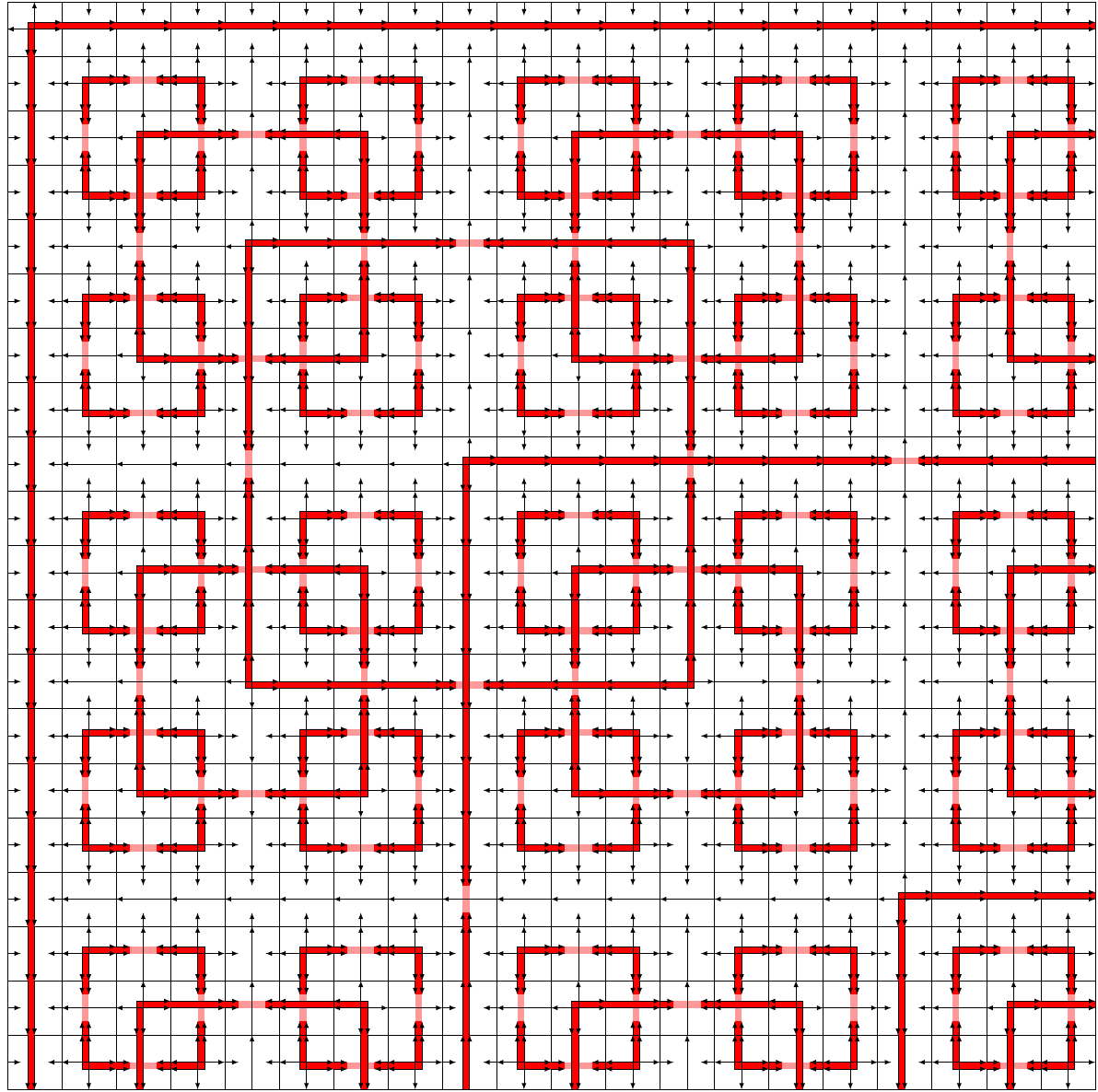}
  \end{center}
  \caption{A valid tiling by Kari's NW-deterministic tileset. The top and left borders determine
  the whole square. The diagonal arrows are not represented but can be easily deduced.
  }
  \label{fig:karitiling}
\end{figure}

Conversely $\left\{\nbper{X}{}\mid X\text{ an SFT}\right\}\subseteq\csharp$: checking whether 
$n\in\strongper{X}$ is $\NP_1$. The number of accepting paths of the Turing machine of the 
proof of lemma~\ref{lem:strong:direct} is exactly the number of periodic points of period $n$.
In order to normalize by $n^d$, we also force this machine to only keep the guessed ``fillings'' of 
the $n^d$ cube that are the smallest among their translates for the lexicographic order.
\end{proof}

\clearpage
\section{1-periodic points in SFTs of dimension 2}

We now go back to bidimensional SFTs, and focus on 1-periodicity. Recall that a point is $1$-periodic 
when it only has colinear vectors of periodicity. We prove theorem~\ref{thm:one}: the sets of 1-periods 
of SFTs are exactly the sets of vectors of $\NN\times\ZZ\setminus\{0\}$ that are in $\NSPACE_1(n)$.

\begin{lemma}
 Let $X$ be an SFT of dimension 2, then $\unaire{\oneper{X}}\in\NSPACE(n)$.
\end{lemma}
\begin{proof}
We have to construct a Turing machine $M$ which on input $\vect v=(m,n) \in\NN\times\ZZ\setminus\{0\}$ 
decides in space $\card{\vect{v}}$ if $\vect v$ is a 1-period.
 We have seen in subsection~\ref{Ss:perpoints} how to construct the graph 
$G_{\vect v}(X)$ and that $\vect v$ is a 1-period iff this graph contains two mutually
accessible cycles. This graph does not fit, however, in space $\card{\vect v}$. The algorithm that 
we will use is similar to the one introduced in lemma~\ref{lem:hori:direct}, it will just need to check
the existence of two different completions that can be ``glued'' together. We suppose without loss of 
generality that $m\geq|n|$ and that $r$ is the radius of the SFT $X$.

\begin{itemize}
 \item Initialize an array $P$ of size $\max(m,n)$ such that $P[i]=1$ for all $i$'s 
 and a boolean $D$ to false.
 \item Nondeterministically sizes $t_1,t_2\leq3m\card{\Sigma}^{4rm}$, the sizes of the two cycles.
 \item Parallely choose $2\times 2r$ horizontal lines $(l_i)_{0\leq i\leq r-1}$ and 
 $(l'_i)_{0\leq i \leq r-1}$ of length $m$, each sequence of lines forms a rectangle of $2r\times m$ 
 symbols of $\Sigma$.
 \item We now do these steps in parallel:
 \begin{itemize}
  \item For all $2r<i\leq t_1$, nondeterministically choose a line $l_1$ of length $m$ and check that 
  there is no  forbidden pattern in the lines $l_i,\dots,l_{i-2r-n}$. For this step, it suffices to 
  keep the $2r$ last lines $l_i,\dots,l_{i-2r}$ and $2r$ symbols on each of the preceeding $n$ lines.
  \item We also nondeterministically choose lines $l'_j$ for $2r<j<t_2$ at the same time and check if 
  they form an invalid pattern. It is important to choose line $l_i$ at the same time as line $l'_j$
  until $\min(t_1,t_2)$ is reached: whenever $l_i$ is different from $l'_i$, assign true to $D$.
  \item At each of these steps, check if there is a periodicity vector $(m',n')$ such that 
  $(m',n')\cdot k = (m,n)$, if it is not the case then $P[k]\leftarrow 0$.
  \item Once the last lines $l_{t_1}$ and $l'_{t_2}$ have been guessed, chck that there is no forbidden
  pattern on the patterns formed by the symbols remembered on lines:
  \begin{itemize}
   \item $l_{t_1-2r-n},\dots,l_{t_1},l_0,\dots,l_{2r}$
   \item $l'_{t_2-2r-n},\dots,l'_{t_2},l'_0,\dots,l'_{2r}$
   \item $l_{t_1-2r-n},\dots,l_{t_1},l'_0,\dots,l'_{2r}$
   \item $l'_{t_2-2r-n},\dots,l'_{t_2},l_0,\dots,l_{2r}$
  \end{itemize}
  This checks that the two cycles found are mutually accessible.
 \end{itemize}
 \item If any of the steps before failed or if there exists $k,m',n'<m$ such that 
  $k\cdot(m',n')=(m,n)$ and $P[k]=1$, or if $D$ is false then reject. Accept otherwise.
\end{itemize}
This algorithm only needs to keep $4rm+2rn$ symbols of $\Sigma$, as well as $P$ and $D$ in memory. 
\end{proof}

\begin{lemma}\label{lem:oneper:struct}
 For any constant $k\in\NN^*$, there exists a 2-dimensional SFT $Y_k$ such that any one periodic point of 
 period $(m,n)$ is formed of $m\times k^{m-1}$ rectangles with marked borders, as in 
 figure~\ref{fig:oneper:skeleton}. Furthermore, $Y_k$ admits as 1-periods any $(m,n)$ such that $0<n<m$.
\end{lemma}

\begin{figure}[h]
 \centering
 \includepicture{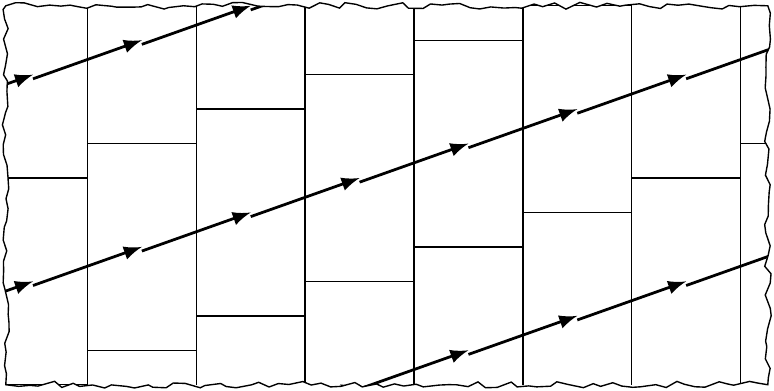}
 \caption{\label{fig:oneper:skeleton} The marked rectangles of the 1-periodic configurations.}
\end{figure}

\begin{proof}
 As in the preceeding proofs, the construction will be done in successive steps, by superimposition 
 of several layers $A,C_c,R,S$ :

\begin{itemize}
  \item Again, the first component $A$ is based on an aperiodic East-deterministic SFT $W$. The 
  alphabet of $A$ is   $\Sigma_A=\left(\Sigma_W\times\{\cdot,\tnoir\}\right)\cup
  \{\leftmost,\rightmost,\betweenrl,\betweenlr\}$, we call the symbols of $\Sigma_W$ 
  the \emph{white} symbols. Again, the other symbols allow to break periodicity, the rules are the
  following:
\begin{itemize}
 \item The rules between the symbols of $\Sigma_W$ remain unchanged.
 \item White symbols may have a \tnoir or another white symbol above, but two \tnoir may not be above/below
 each other.
 \item The constraints on whites ``are transmitted over the \tnoir symbols: removing a line of \tnoir and 
 gluing the white symbols above and below must not produce any forbidden pattern.
 \item The rules between the symbols $\{\tnoir,\leftmost,\rightmost,\betweenrl,\betweenlr\}$ are
 Wang rules.
 \item Only the white sides of symbols $\{\tnoir,\leftmost,\rightmost,\betweenrl,\betweenlr\}$ 
 may touch a white symbol.
\end{itemize}

\begin{figure}
 \begin{center}
 \includepicture{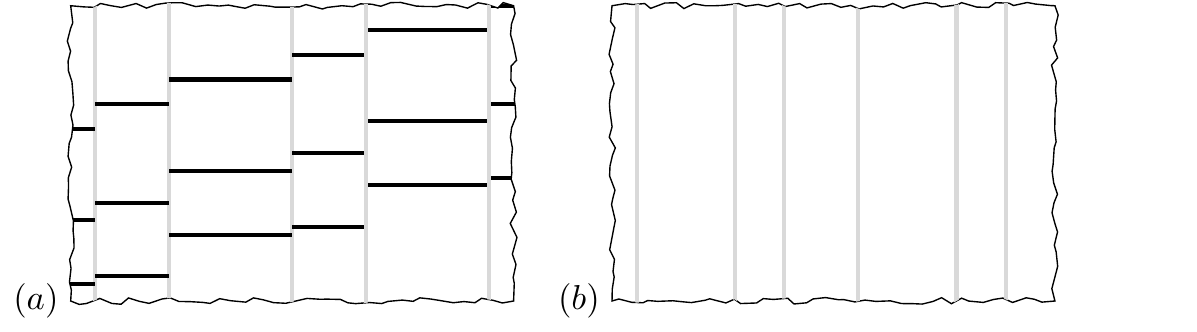}
 \end{center}
 \caption{\label{fig:oneper:compoA}
 There are several possibilities for 1-periodic points of $A$: an infinity of vertical lines $(b)$, or an infinity of vertical lines separated by finite horizontal 
 lines $(a)$.}
\end{figure}

At this stage, the configurations with a periodicity vector necessarily have an infinity of 
vertical lines. Vertical lines may eventually be met by extremities of finite horizontal
lines, see figure~\ref{fig:oneper:compoA}. An infinity of horizontal lines leads to an aperiodic
configuration.

\item Component $C_k$ is a $k$-ary counter, exactly as in lemma~\ref{lem:hori:structure}. Now that this
counter has been added, the points having a periodicity vector necessarily contain an infinity of 
vertical lines, necessarily joined by horizontal lines at distance $k^{n-1}$ when they are distant by $n$.

\item In points with a periodicity vector, component $R$ forces columns formed by the nearest vertical 
lines to all be of the same width, and the offset between horizontal lines of two neighboring columns to always
be the same.
The first is done by first projecting each horizontal line to the left and to
the right until it reaches the next vertical line. Between projections on the same side, we again put a 
a counter $C_k$, this forces the sizes of the two columns to be identical.

To make the offset between horizontal lines constant for all columns, we project two signals of slope 1 on 
both the left and the right of the horizontal lines. These signals propagate normally on white symbols and
cannot touch a vertical line directly: they have to first pass on a horizontal line. As they do this, their 
direction changes and their slope becomes $-1$.  They must then touch the next vertical line at the exact same time as
the projection from two columns to the left/right, at which time the signal stops, see figure~\ref{fig:oneper:proj}.

\begin{figure}[h]
 \centering
 \scalebox{0.75}{\includepicture{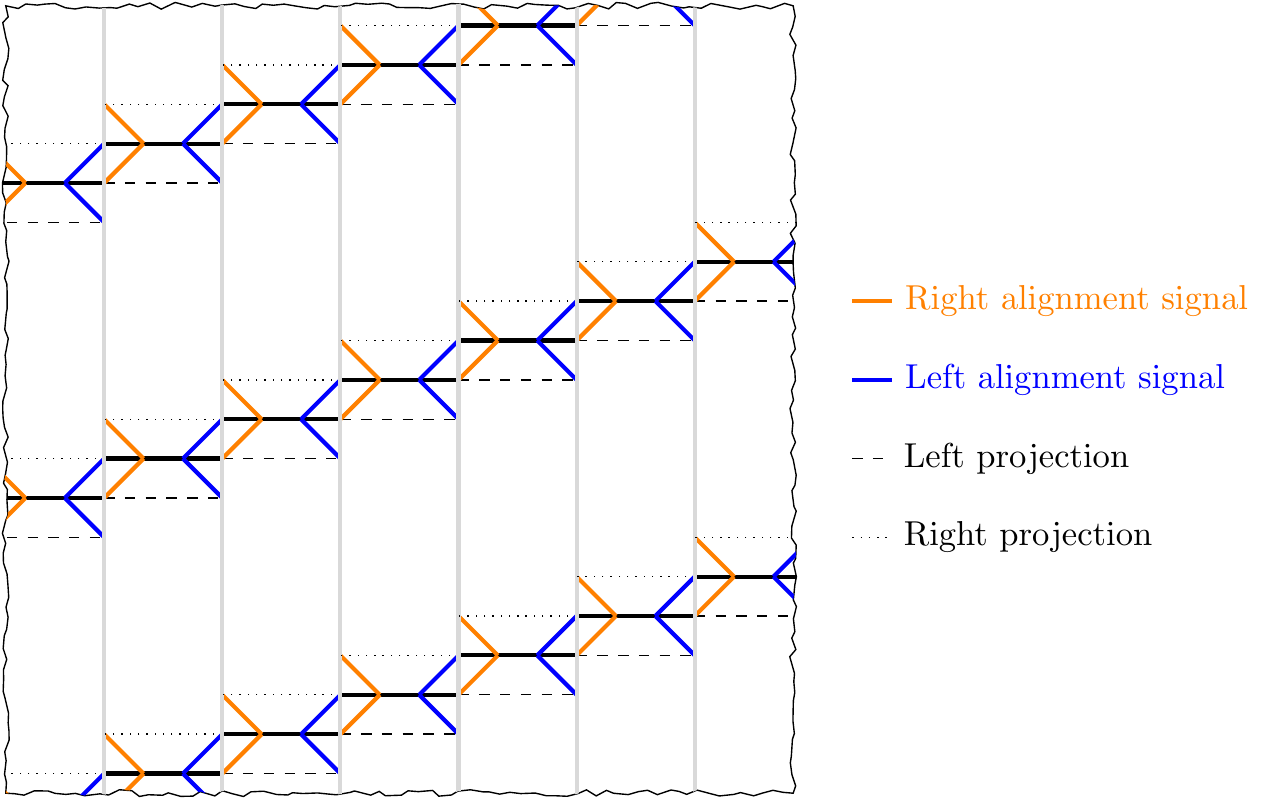}}
 
 \caption{\label{fig:oneper:proj}How component $R$ operates. Each horizontal line sends a projection to
 the next vertical line on each side. On each side, between projections, there is a counter $C_k$. 
 Another signal is send diagonally, with slope 1, it can only touch a horizontal line, at which point 
 it changes its direction and needs to rereach the same vertical line at the same time as the projection 
 from two columns to the left/right.
}
\end{figure}

At this stage, the points with a periodicity vector are necessarily constituted of rectangles of 
identical size, ``translated'' by a vector $(m',n')$ from one column to another. $(m',n')$ is not 
necessarily a periodicity vector, however there must exist $k$ such that $k(m',n')$ is.

\item We now arrange for the width $m$ of the rectangles and their offset $n$ to form the 
periodicity vector: we set the aperiodic background in the columns to be identical.
In order to do that, since $W$ is east-deterministic, we synchronize the first vertical line of white symbols of
a column with the first vertical line of white symbols of the next column, shifting it by the offset.
This is what the last component $S$ does. $S$ is constituted of two sublayers $W'$ whose alphabet is the same
as $W$'s and a second one $S'$ which allows to know which neighboring symbols of $W'$ have to be equal, 
figure~\ref{fig:oneper:transmidecalee} explains $S'$ in sufficient detail to infer the rules. As before, the 
symbols on $W'$ and $A$ to the right of a vertical line need to be identical on both layers.

\begin{figure}[h]
 \centering
 \scalebox{0.25}{\includepicture{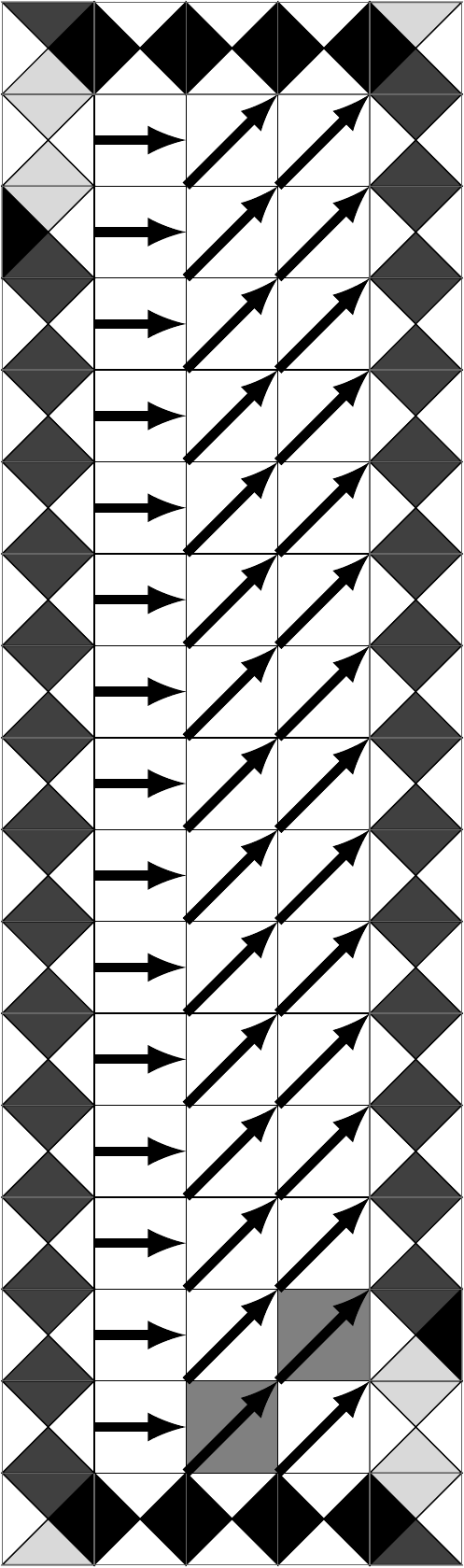}}
 \caption{\label{fig:oneper:transmidecalee} The symbols \protect\trsdiagg are superimposed only 
 to the left alignment signals, and on the bottom line, . The rules for $W'$ are determined by the arrows : a symbol on $W'$ 
 must be equal the symbol on $W'$ to which the arrow points on $S'$.
}
\end{figure}

\item The last step now forces the existence of 1-periodic points, all points with a vector of periodicity at
the end of the previous step were also periodic. To do this, it suffices to add two colors, yellow and blue, to the 
rectangles formed by the vertical lines and horizontal lines. The color of a rectangle is transmitted to its $(m,n)$ 
translate, where $m$ is the width of the rectangle and $n$ the offset between columns.
\end{itemize}

Let us now check that the 1-periodic points are necessarily of the expected shape. Let $x$ be 
a 1-periodic point of 1-period $(m,n)$ with $0<n<m$. Then $x$ necessarily has an infinity of vertical lines due
to $A$. Thanks to $C_k$, two nearest vertical lines at distance $m'$ are joined by horizontal
lines at distance $k^{m'-1}$. $R$ forces all nearest lines to always be spaced by the same distance $m'$. It also
forces the offset between horizontal lines of any two neighboring columns to be some $n'$. Furthermore, $S$ 
forces the aperiodic background of each column to be identical, up to the offset $n'$. So $x$ is periodic 
along $(m',n')$, which is the smallest vector of periodicity, and since $x$ is 1-periodic $(m',n')=(m,n)$.

Conversely, given $(m,n)$ it is easy to exhibit a point with 1-period $(m,n)$: a well formed configuration 
of rectangles of size $m\times k^{m-1}$ where only one of the alignments of rectangles of direction $(m,n)$
is blue and all others are yellow.
\end{proof}

\begin{lemma}\label{Per:lem:oneper:realisation}
 Let $L\in\NN\times\ZZ$ be a  language such that $\unaire{L}\in\NSPACE(n)$, then there exists an SFT $X$ such that
 $L=\oneper{X}$.
\end{lemma}
\begin{proof}
 Let $M$ be a Turing machine recognizing $L$ in time $c^{n}$ for some constant $c$, $n$ being the size of 
 the input. With lemma~\ref{lem:oneper:struct}, it suffices to encode $M$ in the rectangles of $Y_c$ 
 while taking care of synchronizing the nondeterministic transitions between the shifted rectangles, as 
 usual: this allows to realize the subset of $L$ of vectors $(m,n)$ with $0<n,m$. Almost identical 
 constructions allow to do the remaining cases: $0<m<n$, $m=0$,
$n=0$, $n=m$ and $n<0<m$ (there are then two subcases $\card n <\card m$ and $\card m < \card
n$). The disjoint union of all these SFTs form the end SFT realizing $L$ as its 1-periods.
\end{proof}

There are two ways to generalize theorem~\ref{thm:one} to higher dimensions, since in dimension 2 
1-periodicity is both ``$(d-1)$-periodicity'' and 1-periodicity. To obtain a similar characterization,
one would have to consider $(d-1)$-periodicity in higher dimensions because it would again 
come to the problem of whether some graph has two mutually accessible cycles. 1-periodicity 
on the other hand would this time not admit a characterization in terms of some complexity 
class but rather some computability class, as it would be undecidable to decide whether some 
vector is a 1-period or even a periodicity vector.

\section{Periodicity in sofic and effective subshifts}

For periodic and sofic subshifts it is not decidable anymore whether some pattern is admissible, this 
means in particular that it is not decidable anymore whether some $n$ is a period/strong period.

The first thing that we may see however is the following lemma:

\begin{lemma}\label{lem:eff:coRE}
Given an effective/sofic subshift $X$ of dimension $d$
and $d$ vectors $\vect{v_1},\dots,\vect{v_d}$ it is co-recursively enumerable to decide whether
there exists $x\in X$ such that $\Gamma_x = \vect{v_1}\ZZ\oplus\dots\oplus\vect{v_d}\ZZ$.
\end{lemma}
\begin{proof}
 It suffices to guess all fillings of the volume defined by the vectors $\vect{v_1},\dots,\vect{v_d}$,
 we need then to  check if the generated configurations contain some forbidden pattern. For this, we 
 enumerate the forbidden patterns and check whether they appear in every guessed filling (eventually 
 repeated). The machine stops iff all fillings generate a point containing a forbidden pattern.
\end{proof}

The realization counterpart of this lemma for effective subshifts is quite straightforward:

\begin{lemma}\label{lem:eff:realisation}
Let $L$ be a co-recursively enumerable set of rank $d$ sublattices of $\ZZ^d$, then there exists an 
effective subshift $X$ such that $L=\left\{\Gamma_x\mid x\in X\right\}$.
\end{lemma}
\begin{proof}
 It is easy to construct such an SFT with the same methods as before: taking an aperiodic SFT and 
 breaking it with new symbols. The forbidden patterns are all volumes generated by the $d$-uples of 
 vectors of $L$, they also make sure that all such volumes have the same aperiodic background.
\end{proof}

For sofic subshifts we did not manage to obtain a construction allowing us to realize all co-r.e. sets
of rank $d$ lattices. But we managed to realize certain types of lattices, the ``rectangular'' 
lattices, that is to say the lattices of the form $\Gamma=n_1\ZZ\times\dots\times n_d\ZZ$ with 
$n_1,\dots,n_d\in\NN^*$. We will note $\recper{X}$ the sets of $d$-uples such that there 
exists $x\in X$ with $\Gamma_x=n_1\ZZ\times\dots\times n_d\ZZ$. 

\begin{lemma}\label{lem:sofic:realisation}
 Let $L\subseteq {\NN^*}^d$ be a co-recursively enumerable set, there exists a $d$-dimensional sofic 
 subshift $X$ such that $\recper{X}=L$.
\end{lemma}

In order to prove the lemma, we will need the following result:

\begin{theorem}[\citet{AS2010,DRS2010}]\label{thm:AS}
 Let $X$ be a $d$ dimensional subshift, and $X'$ be the $d+1$ dimension subshift 
 obtained by adding a dimension to $X$ and keeping symbols identical on it. $X$ 
is an effective subshift if and only if $X'$ is sofic.
\end{theorem}

This result will be used to generate some sofic subshift by only giving the 
lower dimensional effective subshift corresponding to it.

\begin{proof}[Proof of lemma~\ref{lem:sofic:realisation}]
 We give the proof for dimension $2$, one can easily generalize it for higher dimensions. We have
 a co-r.e. language $L\subseteq {\NN^*}^2$ and we want to construct a sofic subshift $X$ such that 
 $\recper{X}=L$. We now detail the construction realizing the subset $\left\{(n_1,n_2)\mid(n_1,n_2)
 \in L\text{ and } n_1>n_2\right\}$ of $L$. The construction will once again use an aperiodic base
 that will be broken by new symbols, except this time the shape will be determined by layers that
 we put on top and that will not be projected to obtain the sofic subshift.
 
 We describe now the SFT that will be projected:
\begin{itemize} 
 \item The first component $A$ consists of an aperiodic NW-deterministic SFT 
$W$, the white symbols, to which
 we add three symbols \horline, \corner and \breaker.

 The rules are almost the same as for $A'$ of component $A$ of lemma~\ref{lem:strong:struct}: above
 a \breaker or a \corner there may only be a \breaker or a \corner. On the left of a \horline or a 
 \corner, there may only be a \corner or a \horline. For this layer to be periodic, a point needs 
 to have an infinity of lines of \breaker, \horline or \corner. If there is an infinity of lines
 of \breaker and \corner, then they cross with the \corner symbols and do not necessarily form
 rectangles of the same size.

\item The second component $P$ will give the information on the dimensions 
that the rectangles can have and force their width to be identical. To do this, 
we use theorem~\ref{thm:AS} in order to obtain identical lines composed of 
concatenations of words $b^ma^{n-m}$, with $n>m$ and $(n,m)\in L$:
\[
 \dots b^ma^{n-m}b^ma^{n-m}b^ma^{n-m}\dots
\]

These lines can be generated by an effective subshift of dimension $1$, which 
by compactness will contain the uniform lines $\cdots aaa\cdots$ and $\cdots 
bbb\cdots$. So using theorem~\ref{thm:AS} we have an SFT that projects onto the 
shift with identical lines, we may assume that it has a component equal to 
these lines.

We now give the superimposition rules of the symbols: a $b$ appearing after an 
$a$ needs to be superimposed to a \breaker or a \corner, and conversely.

At this stage, component $A$ can only be periodic in two cases: either 
component $P$ is uniform and contains only $a$'s or $b$'s and it has 
an infinity of horizontal lines of \horline, or it contains an infinity of 
vertical lines of \breaker distant by $n$, with $n$ such that there is some 
$m<n$ with $(n,m)\in L$.

\item Component $R$ forces the vertical lines of \horline to appear whenever 
there is an infinity of vertical lines of \breaker. $R$ will be formed of the
following Wang tiles: 
$\left\{\rectdiag,\rectvertleft,\rectvertright,\rectcornerleft, 
\rectcornerright,
\recthoriz,\rectoutside,\rectleft,\rectright\right\}$. 
The superimposition rules are the following:
 \begin{itemize}
  \item \rectvertleft can only be superimposed to \breaker.
  \item \rectcornerleft can only be superimposed to \corner.
  \item \rectcornerright and \rectvertright can only be superimposed to an $a$ 
on the right of a $b$.
  \item \recthoriz can only be superimposed to \horline.
  \item \rectleft,\rectright,\rectdiag can only be superimposed to $b$.
  \item \rectoutside can only be superimposed to $a$.
 \end{itemize}
\begin{figure}
 \centering
 \includepicture{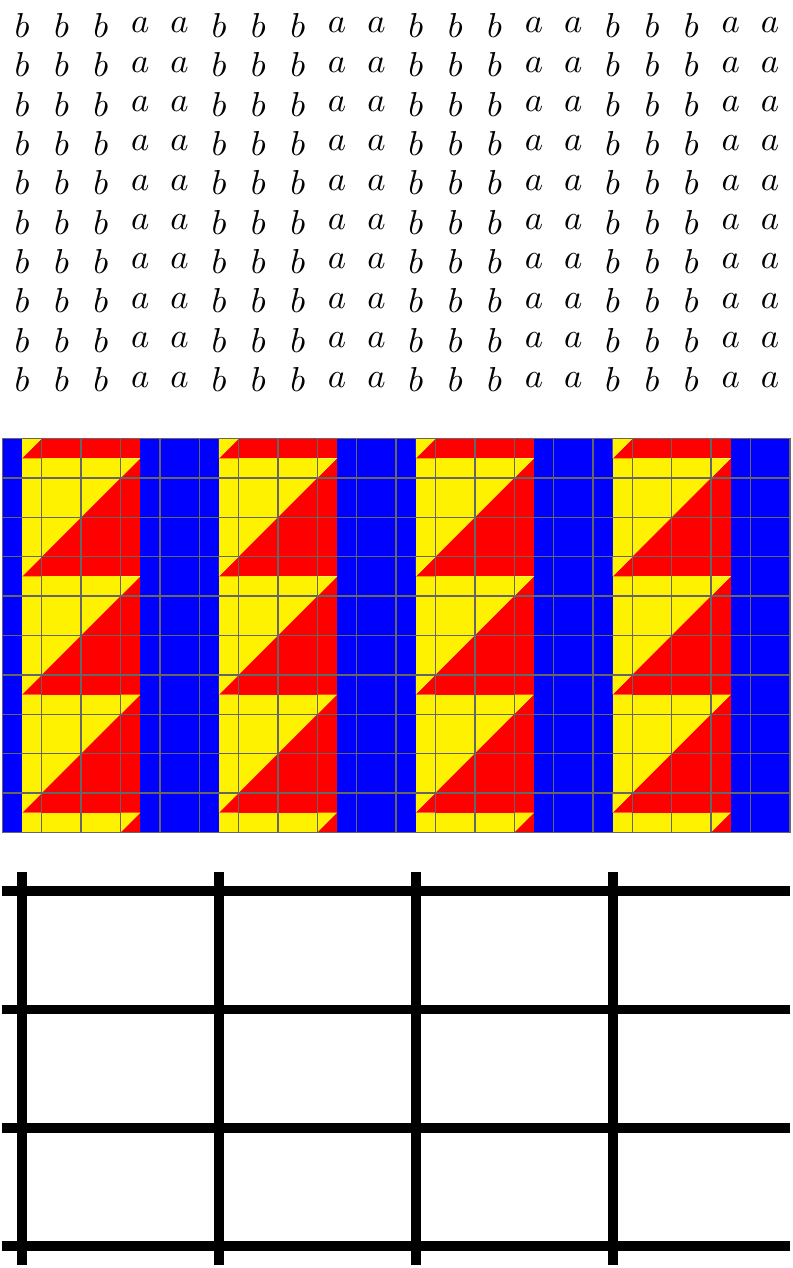}
 \caption{\label{fig:effectifToRect}
 How components $A$, $P$ and $R$ are superimposed: component $P$ is on the top, 
component $R$ is below and on the bottom sits component $A$. The vertical lines 
of \protect\breaker are superimposed to the first symbols $b$ appearing after 
an $a$. Component $R$ allows to adjust the height of the rectangles according 
to the number of $b$'s.}
\end{figure}

Figure~\ref{fig:effectifToRect} shows how component $R$ is superimposed to 
the first components and shapes the rectangles.

The possible cases now when component $A$ is periodic are the following:
there are rectangles, all of the same size $n\times m$, with $n>m$ and 
$(n,m)\in L$, or there is an infinity of horizontal lines of \horline.
In the case when it is formed of rectangles, the vectors $(n,0)$ and $(0,m)$
do not necessarily generate $\Gamma_x$, as the aperiodic background of the
rectangles are not necessarily the same.

\item Component $V$ forces the first column of white symbols of the rectangles 
on component $A$ to be identical. The one-dimensional subshift on the alphabet 
$\Sigma_W\cup\{\#\}$ whose points are of the form:
\[ 
  \cdots \#w_1\dots w_k\#w_1\dots w_k\# w_1\dots w_k\#\cdots
\]
that is to say the periodic points formed by the repetition of some word on
$\Sigma_W$ where each occurence is separated by a $\#$. Again, using 
theorem~\ref{thm:AS} we may use the subshift formed of these lines vertically.
Now, $\#$ can only be superimposed to a \corner or a \horline and no other 
symbol can be superimposed to them. We also force the symbols of $W$ exactly on 
the right of a \breaker to be identical to the ones on $V$. 

\item $H$ is the last component, and is similar to $V$. Consider the effective 
subshift formed by the following points:
\[
  \cdots \#(w_1,t_1)\dots (w_k,t_k)\#(w_1,t_1)\dots (w_k,t_k)\# (w_1,t_1)\dots
(w_k,t_k)\#\cdots 
\]
Where the $w_i$'s are symbols of $\Sigma_W$ and the $t_i$'s are symbols of some 
aperiodic effective subshift $M$ of dimension 1: for example the subshift 
where we forbid the words $awawa$ with $a$ a letter and $w$ a word, which is 
aperiodic and non-empty, see~\citet{Thue,Morse}. 
We keep the rules of $M$ for the words $t_1\dots t_k$ that are not cut by some 
$\#$. So if no $\#$ appears in a point, then the point cannot be periodic. We 
can now use theorem~\ref{thm:AS} and use the subshift formed of these lines 
horizontally.
 
We force the $\#$ symbols to be superimposed to \corner's or \breaker's only
and the symbols of component $A$ appearing on just below a \horline to be 
identical as the ones on $H$.

At this point, if we look the possible periodic points for the restriction to 
component $A$, we have the following:
\begin{itemize}
 \item The points $x$ formed of $n\times m$ rectangles with $n>m$ and $(n,m)\in 
L$, they are periodic and $\Gamma_x=n\ZZ\times m\ZZ$ 
 \item The points formed by an infinity of horizontal lines of \horline, 
without any vertical line of \breaker.
\end{itemize}
\end{itemize}

The end subshift is obtained by projecting component $A$ and $H$ only: 
projecting $H$ allows to get rid of the second class of points, as they not 
have any periodicity vector of the form $(0,m)$, because of the aperiodic 
one-dimensional subshift.
\end{proof}

So we have the following characterizations for sofic and effective subshifts:

\begin{theorem}\label{thm:eff}
 The sets of rank $d$ periodicity lattices of effective subshifts of dimension 
$d$ are exactly the co-r.e. sets of sublattices of $\ZZ^d$.
\end{theorem}

\begin{theorem}\label{thm:sofic}
 For any $L\subseteq {\NN^*}^d$, there exists a $d$-dimensional sofic subshift $X$ such 
that $\recper{X}=L$ if and only if $L$ is co-r.e.
\end{theorem}

\printbibliography
\end{document}